\documentclass[12pt,notitlepage]{article}
\usepackage[utf8]{inputenc}
\usepackage[margin=1in]{geometry}
\usepackage[doublespacing]{setspace}

\usepackage{chngcntr}
\usepackage{comment}
\usepackage{amssymb}
\usepackage{amsmath}
\usepackage{amsthm}
\newtheorem{assumption}{Assumption}
\newtheorem{subassumption}{Assumption\normalfont}[assumption]

\newenvironment{assumption*}
 {\ifnum\value{subassumption}=0 \stepcounter{assumption}\fi\subassumption}
 {\endsubassumption}
\newenvironment{assumption+}[1]
 {\renewcommand{\thesubassumption}{#1}\subassumption}
 {\endsubassumption}
 
\newtheorem{claim}{Proposition}
\newtheorem{lemma}{Lemma}
\newtheorem{theorem}{Theorem}
\newtheorem{definition}{Definition}
\theoremstyle{definition}
\newtheorem{example}{Example}

\usepackage{todonotes}
\usepackage[english]{babel}
\usepackage[authoryear]{natbib}
\usepackage{bbm,graphicx,bm} %
\usepackage[colorlinks = TRUE, allcolors = blue]{hyperref}
\usepackage{float} 
\usepackage{caption}
\usepackage[labelformat=simple,labelfont={}]{subcaption}

\usepackage{enumitem}

\usepackage{booktabs}
\usepackage{graphicx}

\newcommand{\R}{\ensuremath{\mathbb{R}}}

\newcommand{\bbone}{\ensuremath{\mathbbm{1}}}
\newcommand{\E}{\ensuremath{\mathbb{E}}}

\newcommand{\argmin}{\text{argmin} \;\; }

\newcommand{\sep}{\text{sep}}
\newcommand{\cat}{\text{cat}}
\newcommand{\avg}{\text{avg}}

\def\b1{\boldsymbol{1}}

\newcommand{\com}{\text{com}}

\title{Using Multiple Outcomes to Improve the Synthetic Control Method\thanks{We thank Alberto Abadie, Anish Agarwal, David Bruns-Smith, Juanjo Dolado, Bruno Ferman, Brian Jacob, Daniel Lewis, Jesse Rothstein, Kaspar W{\"u}thrich, and Yiqing Xu for helpful discussions on this project. The paper also benefited from the comments of seminar and conference audiences at the ASSA 2024, University of Exeter Business School Advances in DiD Workshop, SED 2024, Polmeth 2024, IMS 2024, JSM 2024. We would like to thank Brian Jacob for his assistance in analyzing data relating to the Flint water crisis. The data used by Jacob and co-authors in \citet{trejo_psychosocial_2021} was provided by the Michigan Education Data Center under restricted use license. For this paper, Jacob ran code we provided to him, which generated the results presented in this paper. This allowed us to examine the value of our proposed approach in a realistic context without having to access restricted data ourselves. Zhenghao Chen provided expert research assistance. Avi Feller and Liyang Sun gratefully acknowledge support from the Institute of Education Sciences, U.S. Department of Education, through
Grant R305D200010. Liyang Sun also acknowledges support from Ayudas Juan de la Cierva Formaci\'on and Grant PID2022-143184NA-I00 funded by MCIU/AEI/ 10.13039/501100011033, and by FEDER, UE. See \href{https://github.com/ebenmichael/augsynth}{\texttt{augsynth}} for the associated \texttt{R} library.
}}
\begin{document}

\date{\today}

\author{Liyang Sun\footnote{Department of Economics, University College London and CEMFI, Email: liyang.sun@ucl.ac.uk}, Eli Ben-Michael\footnote{Department of Statistics \& Data Science and Heinz College of Information Systems \& Public Policy, Carnegie Mellon University.}, and Avi Feller\footnote{Goldman School of Public Policy \& Department of Statistics, University of California, Berkeley.}} 

\maketitle
\centerline{\bf Abstract}%
\smallskip%
\noindent{%
\singlespacing
\small{When there are multiple outcome series of interest, Synthetic Control analyses typically proceed by estimating separate weights for each outcome. In this paper, we instead propose estimating a common set of weights across outcomes, by balancing either a vector of all outcomes or an index or average of them. Under a low-rank factor model, we show that these approaches lead to lower bias bounds than separate weights, and that averaging leads to further gains when the number of outcomes grows. We illustrate this via a re-analysis of the impact of the Flint water crisis on educational outcomes.
}
\smallskip

\noindent\texttt{Key words:}  panel data, synthetic control method, linear factor model 

\smallskip

\noindent\texttt{JEL classification codes:} C13, C21, C23.}
\clearpage 
\section{Introduction}
The synthetic control method (SCM) estimates a treated unit's counterfactual untreated outcome via a weighted average of observed outcomes for untreated units, with weights chosen to match the treated unit's pre-treatment outcomes as closely as possible \citep{AbadieAlbertoDiamond2010}.
In many applications, researchers are interested in multiple outcome series at once, such as both reading and math scores in educational applications \citep[e.g.,][]{trejo_psychosocial_2021} or both low-wage employment and earnings when studying minimum wage changes \citep[e.g.,][]{jardin2017minimum}. Other recent empirical examples with multiple series include \cite{billmeier_assessing_2013,kleven_taxation_2013,bohn_did_2014,pinotti_economic_2015,acemoglu_value_2016,dustmann_labor_2017,cunningham_decriminalizing_2018,kasy2022employing}. There is limited practical guidance for using SCM in this common setting, however, and researchers generally default to estimating separate weights for each outcome.

Like other single-outcome SCM analyses, this separate SCM approach can run into two main challenges for a given number of pre-treatment periods.
In shorter panels, SCM weights can sometimes achieve perfect or near-perfect pre-treatment fit --- but can overfit to idiosyncratic errors, rather than find weights that balance latent factors 
\citep{AbadieAlbertoDiamond2010}. 
At the other extreme, longer panels can mitigate overfitting but are more likely to result in poor pre-treatment fit, which can also introduce bias \citep{ferman2018revisiting, benmichael2021_ascm}.
In addition to these statistical concerns, separate weights for each outcome series can be difficult to interpret, since SCM weights --- including which ``donor units'' have non-zero weight --- typically differ across separate SCM fits.

In this paper, we show that estimating a single set of weights common to multiple outcome series can help address these challenges in cases where the outcomes share a common factor structure, while also offering greater interpretability compared to using separate weights. 
We consider two approaches.
First, following \citet{tian2023synthetic} as well as several recent empirical studies, we find a single set of \emph{concatenated} weights: SCM weights that minimize the imbalance in the concatenated pre-treatment series for all outcomes.
Second, we find a single set of \emph{average} weights: SCM weights that minimize the imbalance in 
a linear combination of pre-treatment outcomes; as the leading case, we focus on imbalance in the average of the pre-treatment outcome series.

Under the assumption that the $K$ different outcome series share a similar factor structure, we derive finite-sample bounds on the bias for these two approaches, as well as bounds when finding separate SCM weights for each outcome series. Unlike previous literature, our bounds apply even if the pre-treatment fit is not perfect.
We show that both the concatenated and averaging approaches reduce potential bias due to overfitting to noise by a factor of $\frac{1}{\sqrt{K}}$ relative to the analysis that considers each outcome separately. We also show that the averaging approach further reduces potential bias due to poor pre-treatment fit by a factor of $\frac{1}{\sqrt{K}}$ relative to \emph{both} the separate and concatenated approaches.
In particular, averaging reduces noise, which both improves pre-treatment fit and reduces bias due to overfitting. 

We next outline considerations for practice, including diagnostics for assessing the important assumption of a common factor structure and for sensitivity to hyperparameters. We conduct a re-analysis of \citet{trejo_psychosocial_2021}, who study the impact of the Flint water crisis on student outcomes in Flint, Michigan. In the Online Appendix, we conduct a Monte Carlo analysis to assess how the common factor structure affects the performance of the concatenation and averaging approaches, and to illustrate that our proposed diagnostics effectively detect the absence of common factors.
Taken together, we argue that --- when multiple outcomes share a common factor structure --- SCM based on averaged outcomes is a reasonable, intrepretable procedure that effectively leverages multiple outcomes for bias reduction.

\paragraph{Related literature.} 
Despite the many empirical examples of SCM with multiple outcomes, there is relatively limited methodological guidance for this setting. 
\citet{Robbins2017} consider this problem in the context of SCM with high-dimensional, granular data and consider different aggregation approaches. 
\citet{amjad_mrsc_2019} introduce the Multi-Dimensional Robust Synthetic Control (mRSC) method, which fits a linear regression using a de-noised matrix of all outcomes concatenated together.

The closest paper to ours is independent work from \citet{tian2023synthetic}, who  consider weights based on concatenated outcomes in the same setting as ours where different outcome series share a similar factor structure \citep[see also][]{tian_three_2021}. The authors derive a bias bound that holds conditional on achieving perfect pre-treatment fit for all outcome series simultaneously. In contrast, we derive bias bounds  that allow for imperfect pre-treatment fit, a more general scenario in applied research. Additionally, our novel analysis demonstrates how averaging can reduce finite sample error relative to concatenated weights.

Finally, we build on an expansive literature on the Synthetic Control Method for single outcomes; see \citet{abadie_using_2021} for a recent review.
In particular, several recent papers propose modifications to SCM to mitigate bias both due to imperfect pre-treatment fit \citep[e.g.,][]{ferman2018revisiting, benmichael2021_ascm} and bias due to overfitting to noise \citep[e.g.,][]{kellogg2021combining}. We complement these papers by highlighting how researchers can also incorporate multiple outcomes to mitigate both sources of bias.

\section{Preliminaries}\label{sec:prelim}

We consider an aggregate panel data setting of $N$ units and $T$ time periods. For each unit $i=1,\dots,N$ and at each time period $t=1,\dots,T$,  we observe $K$ outcomes $Y_{itk}$ where $k=1,\dots,K$. We denote the exposure to a binary treatment by $W_{i}\in\{0,1\}$. We restrict our attention to the case where a single unit receives treatment, and follow the convention that this is the first one, $W_1 = 1$. The remaining $N_0 \equiv N-1$ units are possible controls, often referred to as ``donor units.'' To simplify notation, we limit to one post-treatment observation, $T = T_0 + 1$, though our results are easily extended to larger $T$.

We follow the potential outcomes framework \citep{neyman1923} and denote the potential outcome under treatment $w$ with $Y_{itk}(w)$.
Implicit in our notation is the assumption that there is no interference between units and no anticipation. 
Under this setup, we can write the observed outcomes as $Y_{itk} = (1 - W_{i})Y_{itk}(0) + W_i \bbone\{t \leq T_0\} Y_{itk}(0) + W_i\bbone\{t > T_0\} Y_{itk}(1)$.
The treatment effects of interest are the effects on the $K$ outcomes for the treated unit in the post-treatment period, $\tau_k = Y_{1Tk}(1) - Y_{1Tk}(0)$.
We collect the treatment effects into a vector $\bm{\tau}= (\tau_1,\ldots,\tau_K) \in \R^K$. Since we directly observe $Y_{1Tk}(1) = Y_{1Tk}$ for the treated unit, we focus on imputing the missing counterfactual outcome under control, $Y_{1Tk}(0)$.

Throughout, we will focus on \emph{de-meaned} or \emph{intercept-shifted} weighting estimators, which were introduced in the single outcome setting \citep{Doudchenko2017, ferman2018revisiting} and were adapted to multiple outcomes by \citet{tian2023synthetic}, who argue that outcome-specific demeaning is useful for comparing across outcomes. 
We denote $\bar Y_{i\cdot k} \equiv \frac{1}{T_0}\sum_{t=1}^{T_0} Y_{itk}$ as the pre-treatment average for the $k$\textsuperscript{th} outcome for unit $i$, and $\dot Y_{itk} = Y_{itk} - \bar{Y}_{i\cdot k}$ as the corresponding de-meaned outcome. 
We therefore consider estimators of the form:
\begin{equation}
\label{eq:generic_estimator}
\widehat{Y}_{1Tk}(0) \equiv \bar{Y}_{1\cdot k} + \sum_{W_{i}=0} \gamma_i \dot{Y}_{iTk},
\end{equation}
where $\gamma \in \R^{N-1}$ is a set of weights.
Our paper centers on how to choose the weights $\gamma$.

\section{Leveraging Multiple Outcomes for 
 SCM: Identification}
\label{sec:identification}

In this section we outline assumptions on the data generating process that enable sharing information across multiple outcomes.
We describe necessary and sufficient conditions for there to exist a single set of weights that achieves zero bias across all outcomes simultaneously, and give intuition and examples in terms of linear factor models.

Throughout, we make the following structural assumption on the potential outcomes under control, similar to \cite{athey_matrix_2021}. As in the classic SCM literature \citep{AbadieAlbertoDiamond2010}, the uncertainty arises from the idiosyncratic errors, which are assumed to be strictly exogenous; we condition on deterministic but unknown model components. We assume that correlation across outcomes is captured by model components described later, with idiosyncratic errors independent across outcomes. 
\begin{assumption}\label{assumption:noiseless}
The potential outcome under control is generated as
$$Y_{itk}(0) = \alpha_{ik} + 
\beta_{tk} + L_{itk} + \varepsilon_{itk}$$
where  the  deterministic model component includes unit and time fixed effects $\alpha_{ik}$ and $\beta_{tk}$, with $\sum_{t=1}^T \beta_{tk} = 0$ for all $k$. After incorporating the additive two-way fixed effects, the model component retains a term $L_{itk}$ with $\sum_{i=1}^N L_{itk} = 0$ for all $t, k$ and $\sum_{t=1}^T L_{itk} = 0$ for all $i, k$.
The idiosyncratic errors $\varepsilon_{itk}$  are mean zero, independent of the treatment status
$W_{i}$, and independent across units and outcomes.
\end{assumption}

This setup allows the model component to include $\alpha_{ik}$,  a unit fixed effect specific to outcome $k$.
 We explicitly account for the presence of these fixed effects by de-meaning across pre-treatment periods within each unit's outcome series.
\subsection{Existence of common weights shared across outcomes}
\label{subsec:factor}

To begin, we first characterize the bias of a de-meaned weighting estimator under Assumption~\ref{assumption:noiseless}. For a set of weights $\gamma$ that is independent of the idiosyncratic errors in period $T$, $\widehat{Y}_{1Tk}(0)$ has bias:
\begin{equation}
 \E_{\varepsilon_{T}}\left[Y_{1Tk}(0)  - \widehat{Y}_{1Tk}(0) \right] = \beta_{Tk} \left(1 - \sum_{W_{i}=0} \gamma_i\right) + L_{1Tk} - \sum_{W_{i}=0} \gamma_i L_{iTk},\label{eq:SCM_demean_bias}   
\end{equation}
\noindent where $Y_{1Tk}(0)$ is the $k$\textsuperscript{th} control potential outcome for the treated unit at time $T$. Here the expectation is taken over the idiosyncratic errors in period $T$.

From this we see that weights $\gamma$ will lead to an unbiased estimator for time $t$ and outcome $k$ if (i) the weights sum to one and (ii) the weighted average of the latent $L_{itk}$ for the donor units equals $L_{1tk}$ for the treated unit.
Weights that satisfy these conditions for all time period/outcome pairs would yield an unbiased estimator for every $Y_{1tk}(0)$ simultaneously.
We refer to such weights as \emph{oracle weights} $\gamma^{\ast}$, since they remove the bias due to the presence of the unobserved model components $L_{itk}$. 

\begin{definition}[Oracle Weights]\label{def:oracle}
The oracle weights $\gamma^{\ast}$ solve the following system of $(TK) + 1$ equations
\begin{equation}
\left(\begin{array}{cc}L &\mathbf{1}_N\end{array}\right)'\left(\begin{array}{c}
-1\\
\gamma^{\ast}
\end{array}\right)=\mathbf{0}_{TK},\label{eq:oracle}
\end{equation} 
 where the first row of $L\in\mathbb{R}^{N\times(TK)}$ contains $L_{itk}$ 
 for the treated unit and the remaining rows correspond to control units. 
\end{definition}

We show in Section~\ref{sec:alternative methods} that if such oracle weights exist, we can pool information across outcomes by finding a single set of synthetic control weights that are common to all $K$ outcomes.
Such weights will exist if and only if the underlying matrix of model components $L$ is low rank. We formalize this in the following assumption and proposition.

\begin{assumption*}[Low-rank $L$]\label{claim:low rank}

 The $N\times (TK)$ matrix of model components has reduced rank, and its rank is equal to the $(N-1)\times (TK)$ matrix of model components of the control units $L_{-1}$:
\[
rank(L_{-1})=rank(L) < N-1.
\]
\end{assumption*}

\begin{claim}[Low-rank is sufficient and necessary] \label{claim:factor}
The unconstrained oracle weights $\gamma^{\ast}$ exist iff Assumption~\ref{claim:low rank} holds.

\end{claim}

Finally, even if oracle weights that balance model components across all $K$ outcomes exist, estimating weights can be challenging without further restrictions.
For example, there may be infinitely many solutions to Equation~\eqref{eq:oracle}.
We therefore introduce the following regularity condition %
 that oracle weights \emph{with a bounded norm} exist.

\begin{assumption*}
\label{assumption:oracle}
In addition to Assumption~\ref{claim:low rank},  assume there is a known $C$ such that some oracle weights exist in a set $\mathcal{C}$ where $\|x\|_1 \leq C$ for all $x \in \mathcal{C}$. Denote $\gamma^\ast$ as a solution to Equation~\eqref{eq:oracle} in $\mathcal{C}$.
\end{assumption*}

Below, we will estimate synthetic control weights that are constrained to be in $\mathcal{C}$; Assumption~\ref{assumption:oracle} ensures that this set contains at least some oracle weights, allowing us to compare the synthetic control and oracle weights.
This assumption further ensures that these oracle weights are not too extreme, as measured by the sum of their absolute values.
While we keep the constraint set $\mathcal{C}$ general in our formal development, in practice---and in our empirical analysis below---this constraint set is often the simplex $\mathcal{C} = \Delta^{N_0-1}$, where $C = 1$.
This adds the stronger assumption that there exist oracle weights that are non-negative, and so the model component for the treated unit $L_{1\cdot} \in \R^{TK}$ is contained in the convex hull of the model components for the donor units, $\text{conv}\{L_{2\cdot}, \ldots, L_{N \cdot}\}$.

\subsection{Interpretation for linear factor models}
\label{sec:factor_model_interpretation}

Proposition~\ref{claim:factor} shows that determining whether oracle weights exist is equivalent to determining whether the model component matrix $L$ is low rank.
We now discuss when this assumption is plausible and
how it relates to the more familiar low rank assumptions used in the panel data literature.

To further interpret these restrictions, it is useful to express the model components  $L$ in terms of a linear factor model.
Under Assumption~\ref{claim:low rank}, for $r=rank(L)$ the deterministic model component can be written as a linear factor model,
\begin{equation}
\label{eq:factor-deterministic}
L_{itk}=\boldsymbol{\phi}_{i}\cdot\boldsymbol{\mu}_{tk},
\end{equation}
where $\boldsymbol{\mu}_{tk}\in\mathbb{R}^{r}$ are latent time- and outcome-specific factors and each unit has a vector of time- and outcome-invariant factor loadings
$\boldsymbol{\phi}_{i}\in\mathbb{R}^{r}$.\footnote{This factor structure can be based on a singular value decomposition  $L=UDV'$. Define   $\Upsilon=VD$ . Then we can write $L=U\Upsilon'$ where for $r=rank(L)$,
$\Upsilon\in\mathbb{R}^{(TK)\times r}$ are the latent time-outcome factors
and $U\in\mathbb{R}^{N\times r}$ are the loadings.} Linear factor models like this have been used extensively in the literature to capture a common set of unobserved predictors across different outcomes, e.g., \citet{farias_learning_2019, amjad_mrsc_2019,tian2023synthetic}. 

Proposition~\ref{claim:factor} guarantees that oracle weights exist and that there exists a linear combination of control units' factor loadings that can recover the treated unit's factor loading:  $$\boldsymbol{\phi}_{1}=\sum_{W_{i}=0}\gamma^\ast_i \phi_{i}.$$ 
 
To interpret this factor structure, note that a special case that satisfies Assumption~\ref{claim:low rank}  is where the model component $L_{itk}$ can be decomposed into a common component that is shared across outcomes and an idiosyncratic, outcome-specific component:
\begin{equation}
L_{itk}=\sum_{f=1}^{r_{0}}\phi_{icf}\mu_{tkf}+\sum_{f'=r_{0}+1}^{r_{k}}\phi_{ikf'}\mu_{tkf'} \label{eq:common and idiosyncratic},
\end{equation}
where all loading vectors $\phi_{cf}$ and $\phi_{kf'}$ are orthogonal to each other. 
Let $r_{0}$
denote the  dimension of the factor loadings that are shared across the outcomes. Then we can calculate
 $rank(L)=r_{0}+\sum_{k=1}^{K}(r_{k}-r_{0})$, where there are $r_{0}$  common  factor loadings and $(r_{k}-r_{0})$ idiosyncratic factor loadings for outcome $k$. The factor loadings can be seen as latent feature vectors associated with each unit, which may vary with the outcomes of interest.
The low-rank Assumption~\ref{claim:low rank} then states that 
$r_{0}+\sum_{k=1}^{K}(r_{k}-r_{0})<N-1$.
This can happen when either the number of outcomes $K$ is relatively small or $r_{0}$
is large compared to $r_{k}$ so that there is a high degree of shared information across outcomes.

Importantly, the assumption that the model components are low rank is inherently a substantive one and should be driven by economic theory when possible; we offer some examples next. In Section \ref{subsec: bias discussion}, we suggest some diagnostics for practice, recognizing that the power of these tests is necessarily limited.\footnote{Formal tests for low rank typically require large-sample approximations and stricter assumptions on the errors, which are not imposed here. Moreover, such formal tests can sometimes be misleading. For instance, recent findings show the factor structure of excess bond returns --- and therefore whether these are low rank --- is hard to estimate \citep[see, for example,][]{crump2022factor}.}
\begin{example}[Repeated measurements of the same outcome]
An extreme case is where $Y_{it1}, \ldots Y_{itK}$ are $K$ repeated measurements of the same outcome. In this case 
$\mu_{tk}=\mu_{t}$ for
$k=1,\dots,K$, there are no idiosyncratic terms, and the rank of $L$ is $r_0$. This situation was discussed in \cite{sun2024temporal} in the context of high-frequency measurements of the same outcome.
\end{example}
\begin{example}[Multiple test scores]
    Even with different outcomes, in many empirical settings, such as standardized test scores, there are only a few factors that explain most of the variation across outcomes, so $\sum_{k=1}^K(r_k - r_{0})$ is small and the low-rank assumption is plausible.
     For example, across seven test scores collected by \citet{duflo_dupas_kremer_aer2011}, ``average verbal" and
    ``average math" explain 72\% of the total variation.

\end{example}

\section{Leveraging Multiple Outcomes for 
 SCM: Estimation}\label{sec:alternative methods}

\subsection{Measures of imbalance}\label{subsec:imbalance}

In principle, we would like to find oracle weights that
can recover $L_{1Tk}$ from a weighted average of $L_{2Tk},\ldots,L_{NTk}$ for all $k$.
Since the underlying model components are unobserved, however, we must instead use observed outcomes $Y$ to construct feasible balance measures. 
In the classic (de-meaned) synthetic control method applied separately to each series,  weights are chosen to optimize the pre-treatment fit for a single de-meaned outcome $k$:
\begin{eqnarray*}
\hat{\gamma}^{sep}_k \equiv \underset{{\gamma\in  \mathcal{C}}}{\argmin} q_k^{sep}(\gamma)^2,\ q_k^{sep}(\gamma) & \equiv & \sqrt{\frac{1}{T_{0}}\sum_{t=1}^{T_{0}}\left(\dot{Y}_{1tk} -\sum_{W_{i}=0}\gamma_{i}\dot{Y}_{itk}\right)^{2}} .
\end{eqnarray*}

\noindent We refer to these as \emph{separate weights}, because there is a distinct set of weights to separately estimate the effect for each outcome.

Motivated by the common factor structure, we now consider two alternative balance measures that use information from multiple outcome series. 
First, we consider the \emph{concatenated objective}, which simply concatenates the different outcome series together. This is the pre-treatment fit achieved across all  outcomes and pre-treatment time periods simultaneously. We refer to the set of weights that minimize this objective
as the \emph{concatenated weights}:
\[
\hat{\gamma}^{cat}  \equiv  \underset{\gamma\in  \mathcal{C}}{\argmin} q^{cat}(\gamma)^2,\ q^{cat}(\gamma) \equiv  \sqrt{\frac{1}{T_{0}}\frac{1}{K}\sum_{k=1}^{K}\sum_{t=1}^{T_{0}}\left(\dot Y_{1tk} -\sum_{W_{i}=0}\gamma_{i} \dot Y_{itk}\right)^{2}}.
\]

This objective coincides with the ``multiple-outcome SC estimator'' proposed by \cite{tian2023synthetic}, who also derive a novel bound on the estimation error with perfect pre-treatment fit, i.e., $q^{cat}(\hat{\gamma}^{cat} )=0$. Our focus, however, is on bounding the estimation error when pre-treatment fit is imperfect, and on whether we can reduce this error with an alternative estimator. This alternative is the \emph{averaged objective}, the pre-treatment fit for the average of the  outcomes. We refer to the set of weights that minimize this objective as the \emph{average} weights:
\[
\hat{\gamma}^{avg} \equiv  \underset{\gamma\in  \mathcal{C}}{\argmin} q^{avg}(\gamma)^2,\ q^{avg}(\gamma) \equiv  \sqrt{\frac{1}{T_{0}}\sum_{t=1}^{T_{0}}\left(\frac{1}{K}\sum_{k=1}^{K}\dot Y_{1tk} -\sum_{W_{i}=0}\gamma_{i}\dot Y_{itk}\right)^{2}}.
\]

Note that, for any realization of the data, the pre-treatment fit will be better for the averaged objective than for the concatenated objective,  $q^{avg}(\hat{\gamma}^{avg})\leq q^{cat}(\hat{\gamma}^{cat})$.
This finite-sample improvement in the fit also translates to a smaller upper bound on the bias, as we discuss next.

\subsection{Estimation error}

For any estimated weights $\hat{\gamma}$, the estimation error is
\begin{align*}
 \tau_k - \hat{\tau}_k(\hat{\gamma}) &= \dot{Y}_{1Tk}(0)-  \sum_{W_{i}=0} \hat{\gamma}_i \dot{Y}_{itk} &= \underbrace{L_{1Tk}- \sum_{W_{i}=0} \hat{\gamma}_i L_{iTk}}_{\text{bias = imbalance + overfitting}}\;\;+\;\;\underbrace{\dot \varepsilon_{1Tk}-\sum_{W_{i}=0} \hat{\gamma}_i \dot \varepsilon_{iTk}}_{\text{noise}}. 
\end{align*}
The second term in the decomposition is due to post-treatment idiosyncratic errors and is common across the different approaches for choosing weights.
In Online Appendix~\ref{sec:noise_bound} we show that this term has mean zero and can be controlled if the weights are not extreme.

Our main focus will be the first term, the bias due to inadequately balancing model components, which we denote
$\text{Bias}(\hat{\gamma})$. 
We can decompose this 
into two terms using the linear factor model in~\eqref{eq:factor-deterministic}:
\begin{align}
  \text{Bias}(\hat{\gamma})\equiv L_{1Tk}-\sum_{W_{i}=0}\hat{\gamma}_{i}L_{iTk}
 & =\sum_{t=1}^{T_{0}}\sum_{j=1}^{K}{\omega}_{tj}\left(\dot{Y}_{1tj}-\sum_{W_{i}=0}\hat{\gamma}_{i}\dot{Y}_{itj}\right)\ \ (R_{0}) \label{eq:bias_r0}\\
 & -\sum_{t=1}^{T_{0}}\sum_{j=1}^{K}{\omega}_{tj}\left(\dot{\varepsilon}_{1tj}-\sum_{W_{i}=0}\hat{\gamma}_{i}\dot{\varepsilon}_{itj}\right)\ \ (R_{1}) \label{eq:bias_r1}
\end{align}
where the time and outcome specific terms $\omega_{tj}$ are transformations of the factor values that depend on the specific estimator.

The first term, $R_0$, is bias due to imperfect pre-treatment fit in the pre-treatment outcomes, $\dot{Y}_{itj}$.
The second term, $R_1$, is bias due to overfitting to noise, also known as the \emph{approximation error}.
This arises because the optimization problems minimize imbalance in \emph{observed} pre-treatment outcomes --- noisy realizations of latent factors --- rather than minimizing imbalance in the latent factors themselves.

\subsection{Main result: Bias bounds}
We now turn to our main results. 
Our analysis differs from the existing literature in two key ways. First, the results for synthetic controls with a single outcome from \cite{AbadieAlbertoDiamond2010} and multiple outcomes from \cite{tian2023synthetic} assume perfect pre-treatment fit ($R_0=0$) and provide an upper bound on $R_1$. Instead we derive explicit finite sample upper bounds for $R_0$, accommodating  more general settings with imperfect pre-treatment fit.
Second, we quantify the impact of demeaning with a finite number of pre-treatment time periods $T_0$; this contributes to additional bias \citep[often known as ``Nickell bias'' due to][]{nickell1981biases} but vanishes as $T_0$ grows large. 
\label{subsec: bias bounds}

\subsubsection{Additional assumptions}
To derive finite sample bias bounds, we first place structure on the idiosyncratic errors, assuming they are independent across time and do not have heavy tails. 
\begin{assumption}\label{assumption:SG}
The idiosyncratic errors $\varepsilon_{itk}$ are mean-zero sub-Gaussian random variables with scale parameter $\sigma$, i.e., they satisfy the tail bound $P(|\varepsilon_{itk}| \geq t) \leq 2 \exp\left(-\frac{t^2}{2\sigma^2}\right)$.
\end{assumption}
\noindent This assumption encompasses the setting where the idiosyncratic errors have a larger variance for certain outcomes; in this case the common scale parameter $\sigma$ is the maximum of the outcome-specific scale parameters.

Second, we assume an adequate signal to noise ratio for each outcome separately, for all outcomes jointly, and for the average across outcomes.

\begin{assumption}\label{assumption:singularity}
    Denote $\boldsymbol{\mu}_{tk}\in\mathbb{R}^{r}$ as the time-outcome factors from Equation~\eqref{eq:factor-deterministic} and assume that they are bounded above by $M$.
    Furthermore, denoting $\sigma_\text{min}(A)$ as the smallest singular value of a matrix $A$, assume that (i) $\sigma_{min}\left(\frac{1}{T_{0}}\sum_{t}\mu_{tk}\mu_{tk}'\right) \geq \underbar{\ensuremath{\xi}}^{sep} > 0$ for all outcomes $k=1,\dots,K$; (ii) $\sigma_{min}\left(\frac{1}{T_{0}K}\sum_{tk}\mu_{tk}\mu_{tk}'\right) \geq \underbar{\ensuremath{\xi}}^{cat} > 0$; and (iii) $\sigma_{min}\left(\frac{1}{T_{0}}\sum_{t}\left(\bar{\mu}_{t}\right)\left(\bar{\mu}_{t}\right)'\right) \geq \underbar{\ensuremath{\xi}}^{avg} > 0$ where $\bar{\mu}_{t}=\frac{1}{K}\sum_{k=1}^{K}\mu_{tk}$.
\end{assumption}
\noindent Previous literature introduces similar assumptions to avoid issues of weak identification \citep{AbadieAlbertoDiamond2010}. This additional  assumption precludes settings where averaging removes substantial variation in the latent model components over time. Consider, for example, a setting where the model components for different outcomes vary over time in exactly opposite directions. Here averaging would cancel out any signal from their latent model components, and, as a result, our theoretical guarantees for the average weights would no longer hold. However, we can generally rule out these edge cases by economic reasoning or visual inspection of the co-movement across outcomes.

\subsubsection{Bias bounds}
We now formally state the high-probability bounds on the bias for the three weighting approaches.\footnote{While we would ideally characterize the entire distribution of the bias term, upper bounds provide a clear indication of where the bias is most likely to concentrate. They are also widely used in the extant SCM literature \citep[e.g.,][]{AbadieAlbertoDiamond2010,benmichael2021_ascm}.} These bounds hold with high probability over the noise in all time periods and all outcomes, $\varepsilon_{itk}$, and are derived under the assumption that $\varepsilon_{itk}$ are independent across outcomes.\footnote{With independent errors, additional outcomes provide new but noisy measurement for the latent common factors. With correlated noise, we conjecture that some stationarity condition is necessary for the order of the bounds to hold. In the extreme case of perfectly correlated errors, all approaches should lead to the same result.} We can compare these high-probability bounds for fixed $N$ as the number of time periods $T$ and/or the number of outcomes $K$ grow.

\begin{theorem}
\label{thm:error_bounds} Suppose Assumptions~\ref{assumption:noiseless},~\ref{assumption:oracle}, ~\ref{assumption:SG} and ~\ref{assumption:singularity} hold.  Recall that by construction, the estimated weights satisfy $\|\hat \gamma\|_1 \leq C$ and  Assumption~\ref{assumption:oracle} implies $\|\gamma^\ast\|_1 \leq C$. Let $\tilde \sigma=(1+1/\sqrt{T_0})\sigma$ and $\tilde C = 4(1+\|\gamma^\ast\|_2)$. With high probability, the absolute bias for estimating the treatment effect satisfies the bound
\begin{align*}
    \left| \text{Bias}(\hat{\gamma}_k^{\sep}) \right| &\leq \frac{r_kM^{2}}{\underbar{\ensuremath{\xi}}_k^{\sep}}\left[ \quad\quad \tilde C {\sigma}   +   
    \frac{1}{\sqrt{T_0}}\left(2C\tilde \sigma\sqrt{\log2N_0}\right)
    \right], \\[2em]
  \left| \text{Bias}(\hat{\gamma}^{\cat}) \right| &\leq \frac{rM^{2}}{\underbar{\ensuremath{\xi}}^{\cat}}~\left[  \quad\quad {\tilde C \tilde \sigma}  +    \frac{1}{\sqrt{T_0K}}\left(2C\tilde \sigma\sqrt{\log2N_0}\right)\right], \\[2em]
 \left| \text{Bias}(\hat{\gamma}^{\avg}) \right| &\leq \frac{rM^{2}}{\underbar{\ensuremath{\xi}}^{\avg}}~\left[ \frac{1}{\sqrt{K}}~\tilde C\sigma  + \frac{1}{\sqrt{T_0K}}\left(2C\tilde \sigma\sqrt{\log2N_0}\right)  \right].
\end{align*}

\end{theorem}

The proof for Theorem~\ref{thm:error_bounds} relies on the sharp bound $|\text{Bias}(\hat{\gamma})|=|R_0 - R_1|\leq |R_0|+|R_1|$, which leads us to derive bounds on terms in Equations \eqref{eq:bias_r0} and \eqref{eq:bias_r1} respectively. Different from the previous literature that assumes $R_0=0$, we bound $|R_0|$ by the discrepancy in the objectives between estimated and oracle weights. Table~\ref{tab:bound_rates} gives a high-level overview of these results and shows the leading terms in the bounds, removing terms that do not change with $K$ and $T_0$.

\begin{table}[t]
    \centering
    \begin{tabular}{c c c}
    \toprule
     & Bias due to imperfect fit & Bias due to overfitting\\
    \midrule
     $\hat{\gamma}^{sep}$ & $O(1)$ & $O\left(\frac{1}{\sqrt{T_0}}\right)$\\
     $\hat{\gamma}^{cat}$ & $O(1)$ & $O\left(\frac{1}{\sqrt{T_0K}}\right)$\\
     $\hat{\gamma}^{avg}$ & $O\left(\frac{1}{\sqrt{K}}\right)$ & $O\left(\frac{1}{\sqrt{T_0K}}\right)$\\
     \bottomrule
    \end{tabular}
    \caption{Leading terms in high probability bounds on the bias due to imperfect fit and overfitting in Theorem~\ref{thm:error_bounds}, with $N$ fixed.}
    \label{tab:bound_rates}
\end{table}

For both the separate weights $\hat{\gamma}_k^{sep}$ and the concatenated weights $\hat{\gamma}^{cat}$,
imperfect pre-treatment fit---on outcome $k$ alone for the separate weights, and on all outcomes for the concatenated weights---contributes to bias, regardless of the number of pre-treatment periods or outcomes.
This result is consistent with \cite{ferman2018revisiting} who  show that as $T_0\rightarrow\infty$, 
the separate objective function $q_k^{sep}(\gamma)$
does not converge to the objective minimized by the oracle weights, and therefore remains biased. 
In contrast, the bias due to pre-treatment fit for the average weights will decrease with the number of outcomes $K$.
This is because averaging across outcomes reduces the level of noise in the objective.
With many outcomes, the average will be a good proxy for the underlying model components that themselves can be exactly balanced by the oracle weights.
Averaging therefore allows us to get close to an oracle solution, with low bias due to pre-treatment fit.
This result is also consistent with \cite{ferman2018revisiting} since the variance of the noise decreases to zero as both $K$ and $T_0$ grow.
 
The second component of the bias is the contribution of overfitting to noise.
Mirroring prior results \citep[e.g.,][]{AbadieAlbertoDiamond2010}, we find that the threat of overfitting to noise with separate synthetic control weights will decrease as the number of pre-treatment periods $T_0$ increases --- but remains unchanged as $K$ increases. 
In contrast, the bias from overfitting to noise for both the concatenated and the averaged weights will decrease as the product $T_0 K$ increases, albeit for different reasons.
For the concatenated weights, \citet[Chapter 2]{tian_three_2021} and \cite{tian2023synthetic} also show that the bias is inversely proportional to $T_0K$ and argue that the extra outcomes essentially function as additional time periods. Each time period-outcome pair gives another noisy projection of the underlying latent factors, and finding a single good synthetic control for all of these together limits the threat of overfitting to any particular one.
For averaged weights, averaging across outcomes directly reduces the noise of the objective, as we discuss above. The $T_0$ averaged outcomes will therefore have a standard deviation that is smaller by $\approx \frac{1}{\sqrt{K}}$ than the original outcome series, leading to less noise and less potential for overfitting.\footnote{While we focus on how the bias bound scales with $T_0$ and $K$, it also depends on the number of control units $N_0$ through $\|\gamma^\ast\|_2$. In well-behaved settings, this norm can be bounded by $O(1/\sqrt{N_0})$ \citep{Arkhangelsky2021}. \citet{Ferman_Properties_2021} further shows that SCM bias vanishes when both $N_0,T_0\to\infty$. In such asymptotic settings, we expect limited gains from averaging over concatenation and leave a detailed analysis for future work.}

\subsection{Inference}\label{sec:inference main}
Our main results concern estimation. 
There is no consensus on the best inference procedures for SCM with a single treated unit; below we describe a few procedures and highlight their key assumptions. A popular method is the finite population permutation approach of \cite{AbadieAlbertoDiamond2010} that permutes the treatment assignment and relies on permutation distributions for inference. Since the permutation is across control units, this approach can be applied directly when there are multiple outcomes and any of the aforementioned weights can be  used. This will yield a valid test in some settings, such as if the treatment assignment is  random across units. An alternative approach without random treatment assignment is the conformal inference framework proposed by \cite{chernozhukov2021exact}. In Online Appendix~\ref{sec:inference}, we adapt this to the multiple outcome setting and illustrate it using our empirical application.
We note that this mode of inference requires stronger assumptions to ensure consistent counterfactual estimation and that validity of the resulting hypothesis tests relies on an asymptotic framework with many control units. We discuss these departures from our main setting in detail in Online Appendix~\ref{sec:inference}.

\subsection{Recommendations for Empirical Practice}
\label{subsec: bias discussion}
 
We conclude this section by discussing the practical implications of our results and providing guidance for empirical researchers conducting synthetic control analyses with multiple outcomes. We demonstrate these points in our application.

\paragraph{Assess the low-rank assumption.} 
The key assumption underlying our results is that there are shared common factors across outcomes among the different outcomes (Assumption~\ref{claim:low rank}).
If the outcomes share few common factors and have many idiosyncratic ones, neither the averaged nor the concatenated weights may improve on separate weights.
We therefore recommend evaluating this empirically by examining whether a few singular vectors capture the majority of the total variation across outcomes. 
In addition, under a shared factor structure, weights fit on different combinations of outcomes should yield a reasonable level of fit for any particular outcome series, even if that outcome is excluded from the combination.
Therefore, as a rough diagnostic, we suggest holding out each outcome and inspecting the level of fit of the combined weights estimated using the remaining outcome series. 
Finally, as we discuss in Section \ref{sec:factor_model_interpretation} above, we suggest that applied researchers combine this empirical evidence with substantive knowledge about the outcomes.

\paragraph{Standardize each outcome series.}
In practice, outcomes often differ in scale. While the bias bounds in Theorem~\ref{thm:error_bounds} make no restrictions across outcomes --- the bounds depend only on the maximum variance --- there are practical benefits to standardizing outcomes. First, this tightens the bounds by replacing the maximum scale with the standardized variance. Averaging standardized outcomes prevents outcomes with larger scales from dominating. Standardization also aids in interpretation, for instance allowing us to change the sign of each outcome to follow the convention that ``positive'' has the same semantic meaning for all outcomes (e.g., higher test scores are more desirable). Therefore, we recommend standardizing each outcome series using its pre-treatment standard deviation.\footnote{Standardizing by the \emph{estimated} standard deviation rather than the true, unknown standard deviation may induce a small degree of additional dependence across outcomes at different times. We leave a more thorough analysis to future work.}

\paragraph{Conduct sensitivity analysis based on combined weights.}
In some cases, averaging can remove substantial variation, thereby violating Assumption~\ref{assumption:singularity}. To guard against this, we recommend conducting a sensitivity analysis based on finding \emph{combined weights}: SCM weights that control a linear combination of the two objectives, with weight $\nu$ on the concatenated objective and weight $(1-\nu)$ on the averaged objective. In particular, for any $\nu\in[0,1]$, consider $\hat\gamma^{\com} \in\arg\min_{\gamma\in \mathcal C} \nu^\ast q^{avg}(\gamma) + (1-\nu^\ast) q^{cat}(\gamma)$.
This creates 
an imbalance ``frontier," similar to the approach in \cite{benmichael2022_stag},
where $\nu = 0$ corresponds to the concatenated objective and $\nu = 1$ to the average objective. 
As formalized in Lemma~\ref{lem:combined imbalance} of Online Appendix~\ref{sec:proofs}, if the SCM weights yield good pre-treatment fit on \emph{both} the disaggregated and aggregated outcomes, these weights will also achieve the minimum of the two bounds, and the conformal inference approach described in Online Appendix \ref{sec:inference} will also be valid.
In general, finding the optimal $\nu^\ast$ involves model-derived parameters and is infeasible.  
As a heuristic, we consider $\nu = \sqrt{q^{avg}(\hat\gamma^{cat})}/\sqrt{q^{cat}(\hat\gamma^{cat})}$, which is a ratio that is guaranteed to be in $[0,1]$.

\section{Application: Flint Water Crisis Study}\label{sec:flint_synth}

On April 25, 2014, Flint's residents began receiving drinking water from the Flint River, where the water was both corrosive and improperly treated, causing lead to leach from pipes and exposing about 100,000 residents to contaminated water for over 18 months. Concerns persist a decade later, especially regarding the impact on children, who are highly vulnerable to lead exposure.

To assess this impact, \citet{trejo_psychosocial_2021} conduct several different analyses both across school districts and within Flint. We revisit their cross-district SCM analysis, based on a district-level panel data set for Flint and 54 possible comparison districts in Michigan, viewing the April 2014 change in drinking water as the ``treatment.''
The authors focus on four key educational outcomes: math achievement, reading achievement,
special needs status, and daily attendance; all are aggregated to the annual level from 2007 to 2019. We focus on estimation in the main text and defer additional details and inference to Online Appendix \ref{appendix:flint}.

\citet{trejo_psychosocial_2021} argue that these four outcomes are indicative of (aggregate) student psycho-social outcomes  at the district level, and, consistent with our results in Theorem \ref{thm:error_bounds}, 
fit a common set of (de-meaned) SCM weights based on concatenating these outcome series. 
Here we first inspect the weights for the synthetic controls fit separately on each outcome in Figure~\ref{fig:weights}. This illustrates a challenge in interpreting separate weights, as the ``synthetic Flint'' for each outcome series is a composite of different donor districts with limited overlap in the selected donor units across fits. For example, the synthetic control for math achievement places over 30\% of the weight on the Van Buren Public Schools district, a district that receives little to no weight when fitting a synthetic control for the other three outcomes or their average.

\begin{figure}[tb]
  \begin{centering}
       \includegraphics[width=0.8\textwidth]{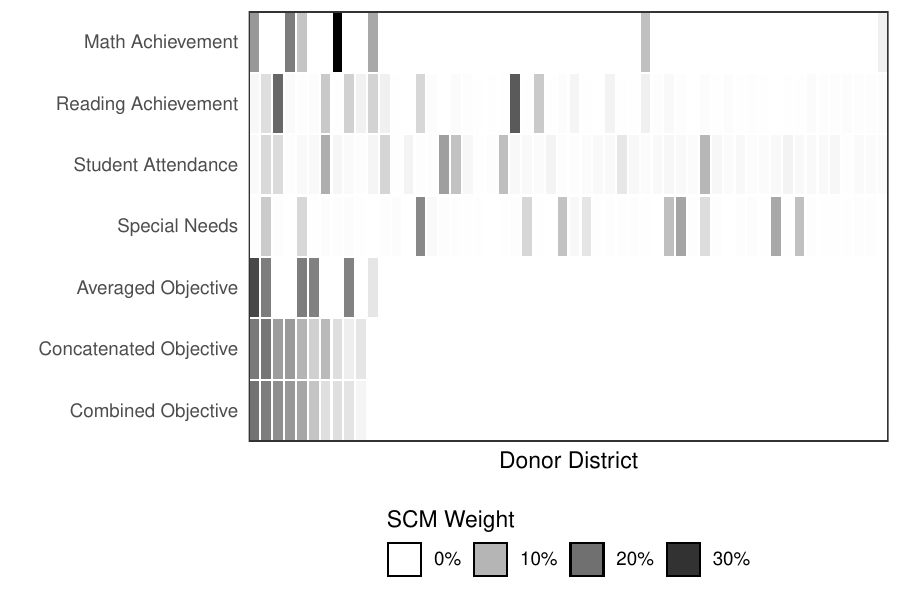}
  \par\end{centering}
  \caption{Separate SCM weight placed on each donor unit using (i) Student Attendance, (ii) Special Needs, (iii) Reading Achievement, and (iv) Math Achievement separately as outcomes, along with the weights solving the concatenated and averaged objectives. %
  The top 5 districts accounting for over 75\% of the weights solving the combined objective are: (i) Dowagiac Union; (ii) Oak Park School; (iii) Lincoln Consolidated; (iv) Hamtramck; and (v) Houghton Lake. These are similar to the weights found by \cite{trejo_psychosocial_2021} (see Online Appendix Table~\ref{tab:combined_weights} for the full list of weights).}\label{fig:weights}
  \end{figure}

Next, we assess whether the observed data are consistent with the low-rank factor model discussed in Section \ref{subsec:factor}.
To do so, we examine the $N \times T_0K$ matrix of (de-meaned and standardized) pre-treatment outcomes, where $N = 54$, $T_0 = 8$, and $K = 4$. The top 10 singular vectors capture over 80\% of the total variation, which is consistent with a low-rank model component and the existence of corresponding oracle weights.
As a second check, we evaluate the ``held out'' pre-treatment fit on each outcome series when we remove that outcome from the combined concatenated and averaging objectives.
Online Appendix Figure~\ref{fig:pre_fit} shows the mean square prediction error relative to uniform weighting. We find that the combined approach yields reasonable held-out fit for math, reading, and special needs outcomes. However, we find worse held-out fit for student attendance; further inspection finds that this is greatly improved after excluding special needs as an outcome.
Taken together, this suggests that while a common factor structure could be reasonable for math and reading achievement, there may be idiosyncrasies between the outcome series for student attendance and special needs.\footnote{
In Online Appendix Figure~\ref{fig:flint_nosped}, we consider analyzing the impact on special needs separately from the other three outcomes, consistent with the robustness checks in \citet{trejo_psychosocial_2021} and this diagnostic check. The results are broadly similar.
}

In addition to the concatenated SCM weights in the original application, we now also consider separate and average SCM weights and the combined approach that solves a weighted average of the concatenated and averaged objectives, with a heuristic weight of $\nu = 0.47$ as shown in Online Appendix Figure~\ref{fig:frontier}. Figure~\ref{fig:gap_flint_fixedeff} shows the SCM gap plots --- i.e. the differences between the observed outcomes for Flint and the counterfactual outcomes imputed by the synthetic control --- for these four sets of weights.
The separate SCM weights achieve close to perfect fit in the pretreatment period, suggesting potential bias due to overfitting to noise, as we discuss in Section \ref{subsec: bias discussion}.
By contrast, the concatenated and average SCM weights and the combined approach do not lead to near-perfect pre-treatment fit, though the fit is still reasonably good. 

The results largely replicate those in \citet{trejo_psychosocial_2021}. Both the averaged and concatenated weights as well as the combined approach estimate a deterioration of math test scores following the Flint water crisis, with little change in reading test scores and student attendance. All three sets of weights also find an increase in the proportion of students with special needs, though the magnitude is smaller for the averaged weights.

\begin{figure}[t]
\begin{centering}
     \includegraphics[width=0.7\textwidth]{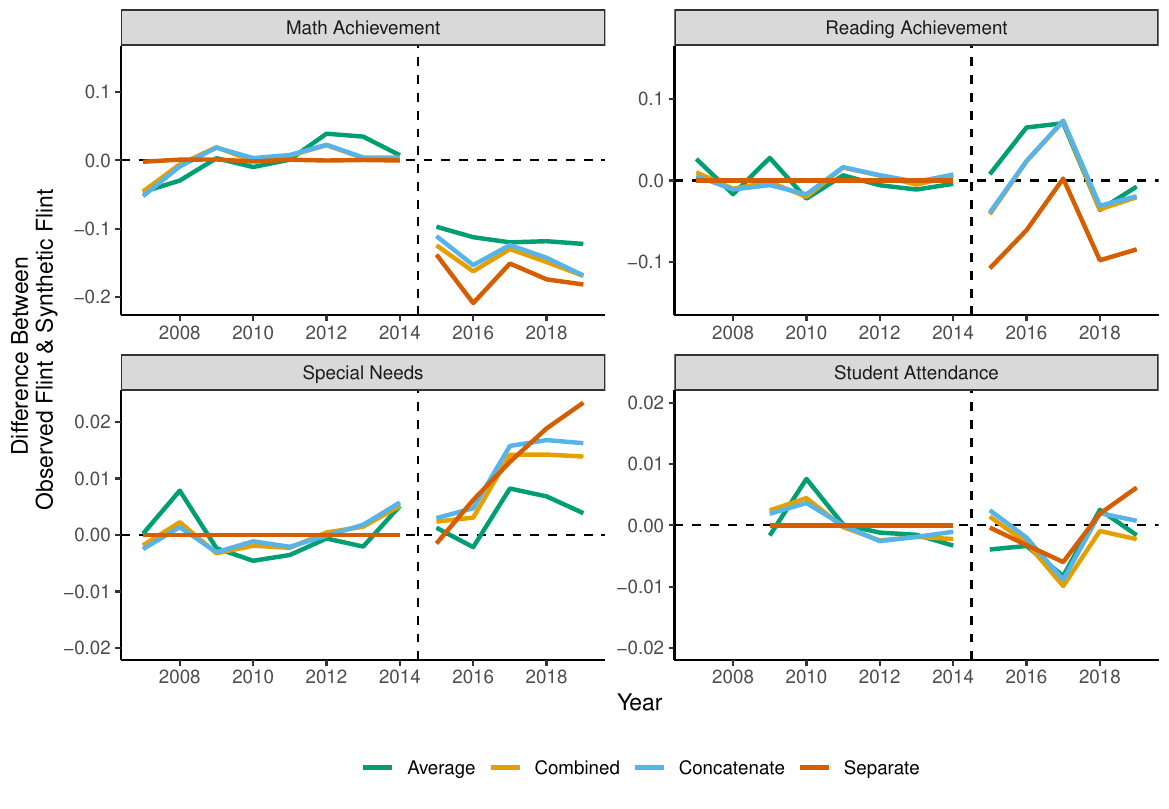}
\par\end{centering}
\caption{Point estimates for the effect of the Flint water crisis using separate weights, concatenated weights, and average weights, as well as the combined weights (setting $\nu = 0.47$). The separate SCM weights yield essentially perfect pre-treatment fit for all four outcomes.}\label{fig:gap_flint_fixedeff}
\end{figure}

\section{Conclusion}
\label{sec:conclusion}
SCM is a popular approach for estimating policy impacts at the aggregate level, such as the school district- or state-level. This approach, however, can be susceptible to bias due to poor pre-treatment fit or to overfitting to idiosyncratic errors. By incorporating multiple outcome series into the SCM framework and under the assumption that the multiple outcomes share a similar factor structure, this paper proposes approaches that address these challenges and that yield more interpretable weights.

There are several directions for future work. The most immediate is alternatives for SCM with multiple outcomes when the common factor structure might not in fact hold.  We could consider approaches that average or borrow strength across multiple model types, including from hierarchical Bayesian models \citep[see, for example,][]{ben2023estimating}, from tensor completion following \citet{agarwal2023synthetic} or from an instrumental variable approach following \citet{shi_proximal_2021} and \citet{fry_method_2023}.
In addition, even when the common factor structure does hold, there may be more efficient ways to create an index of outcomes beyond a simple average, including dimension reduction approaches such as PCA or reduced rank regression.
Finally, leveraging multiple outcomes alone might not be enough to mitigate SCM bias. Following \citet{benmichael2021_ascm} and \citet{Arkhangelsky2021}, we could consider augmenting common SCM weights with either a common outcome model or separate models for each outcome series. 

 \clearpage
\bibliography{refs}
\bibliographystyle{aer}

\clearpage
\appendix
\section*{Online Appendix}
\counterwithin{figure}{section}
\counterwithin{table}{section}

\section{Inference}
\label{sec:inference}

\subsection{Overview}
There is a large and growing literature on inference for the synthetic control method and variants. Here we adapt the conformal inference approach of
\cite{chernozhukov2021exact} to the setting of multiple outcomes.
To do so, we focus on a sharp null hypothesis about the effects on the $K$ different outcomes simultaneously, $H_0:\boldsymbol{\tau}=\boldsymbol{\tau}_0$, with $\boldsymbol{\tau} \in \R^K$. For example, if $\boldsymbol{\tau}_0=\bold{0}_K$ we are interested in testing whether the treatment effect is zero for all outcomes. 

The conformal inference approach proceeds as follows: 

    \begin{enumerate}
        \item Enforce the null hypothesis by creating adjusted post-treatment outcomes for the treated unit $\tilde Y_{1Tk} = Y_{1Tk}-\tau_{0k}$.
         \item  Augment the original data set to include the post-treatment time period $T$, with the adjusted outcomes $\tilde Y_{1Tk}$; use the concatenated, averaged, or combined objective function to obtain weights $\hat \gamma(\boldsymbol{\tau}_0)$
         \item Compute the adjusted residual $\hat u_{tk} =  Y_{1tk} - \sum _{W_i=0}\hat \gamma_i(\boldsymbol{\tau}_0) Y_{itk} $  and $\hat u_{Tk} = \tilde Y_{1Tk} - \sum _{W_i=0}\hat \gamma_i(\boldsymbol{\tau}_0) Y_{iTk} $ and form the test statistic:
         \begin{equation}
             S_q(\hat u_t) = \left (\frac{1}{\sqrt{K}} \sum_{k=1}^{K} |\hat u_{tk}|^q \right)^{1/q}\label{eq:test stat}
         \end{equation}
         where the choice of the norm $q$ maps to power against different alternatives. For instance, if the treatment has a large effect for only few outcomes, choosing $q=\infty$ yields high power. On the other hand, if the treatment effect has similar magnitude across all outcomes, then setting $q=1$ or $q=2$ yields good power.  In practice, we set $q=1$.
         \item Compute a  $p$-value  by assessing whether the test statistic associated with the post-treatment period ``conforms'' with the distribution of the test statistic associated with pre-treatment periods:
        \begin{equation}
             \hat{p}(\boldsymbol{\tau}_0) = \frac{1}{T} \sum_{t=1}^{T_0} \b1 \left\{S_q(\hat u_T) \leq S_q(\hat u_t) \right\} + \frac{1}{T}.  \label{eq:p value}
         \end{equation}
    \end{enumerate}

\citet{chernozhukov2021exact} show that in an asymptotic setting with $T$ (and $N$) growing, this conformal inference procedure will be valid for estimation methods that are consistent.
In particular, they show that the test~\eqref{eq:p value} has approximately correct size; the difference between actual size and nominal size vanishes as $T_0,N \rightarrow \infty$.
In  the next section, we discuss technical sufficient conditions for consistency, closely following \citet{chernozhukov2021exact} and departing from the finite-sample analysis that is our main focus here.

To construct the confidence set for the treatment effect of different outcomes, we collect the values of $\boldsymbol{\tau}_0$ for which test~\eqref{eq:p value} does not reject.  We can then project the confidence set onto each outcome to form a conservative confidence interval.

Finally, an alternative approach is to focus on testing the average effect across the $K$ outcomes, $\frac{1}{K}\sum_{k=1}^{K}\tau_k$, with outcomes appropriately scaled so that positive and negative effects have the similar semantic meanings across outcomes.
This setting returns to the scalar setting considered by \cite{chernozhukov2021exact}, where the estimates are based on the average weights $\hat \gamma^{avg}$,
and so for inference on the average we can follow their procedure exactly.

\subsection{Technical details regarding inference}
\label{sec:technical_inf}
In this section we provide additional technical details for the approximate validity of the conformal inference procedure proposed by \citet{chernozhukov2021exact} with averaged weights.
To do so, we will consider an asymptotic setting with both $N$ and $T$ growing, and make a variation of the structural Assumption~\ref{assumption:noiseless}
and Assumption~\ref{assumption:oracle} that constrained oracle weights exist.

\begin{assumption}
\label{a:conformal}
The de-meaned potential outcome under control for the treated unit's $k$th outcome at time $t$ is
\[
\dot{Y}_{itk}(0) = \sum_{W_i = 0} \gamma_i^\ast \dot{Y}_{itk} + u_{tk},
\]
for some set of oracle weights $\gamma^\ast \in \mathcal{C}$, where for a given $k$ the noise terms $u_{1k},\ldots,u_{Tk}$ are stationary, strongly mixing, with a bounded sum of mixing coefficients bounded, and satisfy $\E[u_{tk} Y_{itk}] = 0$ for all $W_i = 0$.
\end{assumption}
As in the previous assumptions, Assumption~\ref{a:conformal} also assumes the existence of oracle weights $\gamma^\ast$ shared across all outcomes, though they are defined slightly differently.
Directly applying Theorem 1 in \citet{chernozhukov2021exact}, the conformal inference procedure in Section~\ref{sec:inference} using a set of weights $\hat{\gamma}$, will be asymptotically valid if $\sum_{W_i = 0}\hat{\gamma}_i \dot{Y}_{itk}$ is a consistent estimator for $\sum_{W_i = 0}\gamma^\ast_i \dot{Y}_{itk}$, when we include the post-treatment period $T$ when estimating the weights.

Next, we list sufficient assumptions for this type of consistency using the average weights $\hat{\gamma}^\text{avg}$,
consistency with the concatenated weights $\hat{\gamma}^\text{cat}$ can be established in an analogous matter. 
In these assumptions, we define $\bar{u}_t = \frac{1}{K} \sum_{k=1}^K u_{tk}$ and $\dot{\bar{Y}}_{it\cdot} = \frac{1}{K}\sum_{k=1}^K \dot{Y}_{itk}$.

\begin{assumption}
    \label{a:conformal_avg}
    \mbox{}\\
    \begin{enumerate}[label = (\roman*)]
        \item There exist constants $c_1,c_2 > 0$ such that $\E[\left(\dot{\bar{Y}}_{it\cdot} \bar{u}_t\right)^2 \geq c_1$ and $\E[\left|\bar{Y}_{it\cdot} \bar{u}_t\right|^3] \leq c_2$ for any $i$ such that $W_i = 0$ and $t = 1,\ldots,T$.
        \item For each $i$ such that $W_i = 0$, the sequence $\{\dot{\bar{Y}}_{it\cdot} \bar{u}_t\}$ is $\beta$-mixing and the $\beta$-mixing coefficient satisfies $\beta(t) \leq a_1 \exp\left(-a_2 t^\tau\right)$, where $a_1,a_2,\tau > 0$.
        \item There exists a constant $c_3 > 0$ such that $\max_{i: W_i = 0}\sum_{t=1}^T \dot{\bar{Y}}_{it\cdot}^2\bar{u}_t^2 \leq c_3 T$ with probability $1 - o(1)$.
        \item $\log N = o\left(T^\frac{4\tau}{3\tau + 4}\right)$
        \item \label{a:new} There exists a sequence $\ell_T > 0$ such that $\frac{\ell_T M [\log(\min\{T, N-1\})]^\frac{1+\tau}{2\tau}}{\sqrt{T}} \to 0$, 
        \[
            \left(\sum_{W_i = 0} \dot{Y}_{iTk} \delta_i \right)^2 \leq \ell_T \frac{1}{T}\sum_{t=1}^T \left(\sum_{W_i = 0} \dot{\bar{Y}}_{it\cdot} \delta_i\right)^2,
        \]
        and
        \[
          \frac{1}{T}\sum_{t=1}^T\left(\dot{Y}_{itk}\delta_i\right)^2 \leq \ell_T \frac{1}{T}\sum_{t=1}^T\left(\dot{\bar{Y}}_{it\cdot}\delta_i\right)^2
        \]
        for all $\gamma^\ast + \delta \in \mathcal{C}$, all $k =1,\ldots,K$  with probability $1 - o(1)$.
    \end{enumerate}
\end{assumption}

Assumption~\ref{a:conformal_avg} follows the technical assumptions in the proof of Lemma 1 in \citet{chernozhukov2021exact} with two modifications. First, we place assumptions on the noise values averaged across outcomes, $\bar{u}_1,\ldots\bar{u}_T$ rather than the outcome-specific noise values because we are working with the averaged estimator.
Second, Assumption~\ref{a:conformal_avg}\ref{a:new} modifies Assumption (6) in the proof of Lemma 1 in \citet{chernozhukov2021exact} to link consistent prediction of the average of the de-meaned outcomes to consistent prediction for any individual outcome. This assumption is related to Assumption~\ref{assumption:oracle}. If there is a common factor structure across outcomes, then we have the link
\begin{align*}
\sum_{W_i = 0} \dot{Y}_{itk}\delta_i & = \mu_{tk} \cdot \sum_{W_i = 0} \phi_i \delta_i + \sum_{W_i = 0} \dot{\varepsilon}_{itk} \delta_i\\
& =  \mu_{Tk}\cdot\left(\sum_{t}\left(\bar{\mu}_{t}\right)\left(\bar{\mu}_{t}\right)'\right)^{-1} \bar{\mu}_{t} \sum_{W_i = 0} \dot{\bar{Y}}{it \cdot} \delta_i +  \mu_{Tk}\cdot\left(\sum_{t}\left(\bar{\mu}_{t}\right)\left(\bar{\mu}_{t}\right)'\right)^{-1} \bar{\mu}_{t} \sum_{W_i = 0} \dot{\bar{\varepsilon}}_{it \cdot} \delta_i + \sum_{W_i = 0} \dot{\varepsilon}_{itk} \delta_i.
\end{align*}
So, if common oracle weights exist, Assumption~\ref{a:conformal_avg}\ref{a:new} amounts to an assumption on the noise terms.

Under these assumptions, we have a direct analog to Lemma 1 in \citet{chernozhukov2021exact} that is a direct consequence. We state it here for completeness.
\begin{lemma}
    \label{lem:conformal}
    Let $\hat{\gamma}^\text{avg}$ solve $\min_{\gamma \in \mathcal{C}} q^\text{avg}(\gamma)^2$, including the post treatment outcome $T$. Under Assumptions~\ref{a:conformal} and \ref{a:conformal_avg}, $\hat{\gamma}^\text{avg}$ satisfies the consistency properties required for Theorem 1 in \citet{chernozhukov2021exact}, namely,
    \[
        \frac{1}{T}\sum_{t = 1}^T \left(\dot{Y}_{itk} \left(\hat{\gamma}_i - \gamma_i^\ast\right)\right)^2 = o_p(1)
    \]
    and 
    \[
        \dot{Y}_{iTk} \left(\hat{\gamma}_i - \gamma_i^\ast\right) = o_p(1).
    \]
\end{lemma}
\begin{proof}[Proof of Lemma~\ref{lem:conformal}]
    First, we can directly apply the claim from the proof of Lemma 1 and Lemma H.8 in \citet{chernozhukov2021exact} to state that there exists a constant $M >0$ such that
    \[
        \frac{1}{T}\sum_{t=1}^T\left(\dot{\bar{Y}}_{it\cdot}(\hat{\gamma}_i^\text{avg} -\gamma^\ast_i)\right)^2 \leq \frac{M [\log(\min\{T, N-1\})]^\frac{1+\tau}{2\tau}}{\sqrt{T}}
    \]
    with probability $1 - o(1)$. Now from Assumption~\ref{a:conformal_avg}\ref{a:new}, $\frac{\ell_T M [\log(\min\{T, N-1\})]^\frac{1+\tau}{2\tau}}{\sqrt{T}} = o(1)$, which completes the proof.
\end{proof}

\section{Auxillary lemmas and proofs}
\subsection{Error bounds for the oracle imbalance}

The bias due to imbalance in observed demeaned outcomes depends crucially on the measure of imbalance we choose to minimize. We upper bound the imbalance using the estimated weights with the imbalance when using oracle weights, which we refer to as oracle imbalance. For example, we argue the oracle imbalance for the objective function of the SCM satisfies a form of concentration inequality:
\begin{eqnarray*}
q^{sep}(\gamma^{\ast}) & = & \sqrt{\frac{1}{T_{0}} \sum_{t=1}^{T_{0}}\left(\dot{\varepsilon}_{1tj} -\sum_{W_{i}=0}\gamma^{\ast}_{i}\dot{\varepsilon}_{itj}\right)^{2}}
\end{eqnarray*}
At first glance, the imbalance is the L2 norm of the vector of demeaned errors.  The challenge is that the demeaned errors $\dot{\varepsilon}_{itj}$ are correlated over time due to demeaning.

We prove a general upper bound on the oracle imbalance in Lemma~\ref{claim:L2_demeaned_errors} that allow us to decompose the imbalance into the L2 norm of errors and the L2 norm of the average of errors. Lemma~\ref{claim:better_qavg_errors} presents the intermediate concentration inequality for the L2 norm of errors. 
 Finally, building on Lemma~\ref{claim:L2_demeaned_errors} and~\ref{claim:better_qavg_errors}, Lemma~\ref{claim:better_qavg} inspects the numerical properties for the pre-treatment fits achievable by the oracle weights.  Unless otherwise noted, all results hold under Assumptions~\ref{assumption:noiseless},~\ref{assumption:oracle},~\ref{assumption:SG}.

\begin{lemma}[L2 norm of demeaned errors]\label{claim:L2_demeaned_errors}
Under the oracle weights, we have the following upper bounds for the oracle imbalance
\begin{eqnarray*}
q^{cat}(\gamma^{\ast}) & \leq & \sqrt{2\cdot \frac{1}{T_{0}}\frac{1}{K}\sum_{k=1}^{K}\sum_{t=1}^{T_{0}}\left(\varepsilon_{1tk} -\sum_{W_{i}=0}\gamma^{\ast}_{i}\varepsilon_{itk}\right)^{2}} + \sqrt{\frac{2}{K}\sum_{k=1}^{K}\left(\bar \varepsilon_{1\cdot k} -\sum_{W_{i}=0}\gamma^{\ast}_{i}\bar \varepsilon_{i\cdot k}\right)^{2}} \\
q^{avg}(\gamma^{\ast}) & \leq & \sqrt{\frac{2}{T_{0}}\sum_{t=1}^{T_{0}}\left(\frac{1}{K}\sum_{k=1}^{K}\varepsilon_{1tk} -\sum_{W_{i}=0}\gamma^{\ast}_{i}\varepsilon_{itk}\right)^{2}}+\sqrt{2\left(\frac{1}{K}\sum_{k=1}^{K}\bar \varepsilon_{1\cdot k} -\sum_{W_{i}=0}\gamma^{\ast}_{i}\bar \varepsilon_{i\cdot k}\right)^{2}}\\
q^{sep}(\gamma^{\ast}) & \leq & \sqrt{\frac{2}{T_{0}} \sum_{t=1}^{T_{0}}\left({\varepsilon}_{1tj} -\sum_{W_{i}=0}\gamma^{\ast}_{i}{\varepsilon}_{itj}\right)^{2}} + \sqrt{ 2\left(\bar{\varepsilon}_{1\cdot j} -\sum_{W_{i}=0}\gamma^{\ast}_{i}\bar{\varepsilon}_{i\cdot j}\right)^{2}}.
\end{eqnarray*}
\end{lemma}

\begin{proof}[Proof of Lemma~\ref{claim:L2_demeaned_errors}]
Note the following algebraic inequality 
\begin{eqnarray*}
\left(\dot{\varepsilon}_{1tj} -\sum_{W_{i}=0}\gamma^{\ast}_{i}\dot{\varepsilon}_{itj}\right)^{2} & = & \left(\varepsilon_{1tj} -\sum_{W_{i}=0}\gamma^{\ast}_{i}\varepsilon_{itj} - \left(\bar{\varepsilon}_{1\cdot j} -\sum_{W_{i}=0}\gamma^{\ast}_{i}\bar{\varepsilon}_{i\cdot j}\right) \right)^{2} \\
& \leq &  2\left(\varepsilon_{1tj} -\sum_{W_{i}=0}\gamma^{\ast}_{i}\varepsilon_{itj} \right)^{2} +
2\left(\bar{\varepsilon}_{1\cdot j} -\sum_{W_{i}=0}\gamma^{\ast}_{i}\bar{\varepsilon}_{i\cdot j} \right)^{2}.
\end{eqnarray*}
For brevity, we only prove the upper bound for $q^{sep}(\gamma^{\ast})$ as the other two upper bounds can be shown similarly. 
\begin{eqnarray*}
q^{sep}(\gamma^{\ast}) & \leq & \sqrt{\frac{2}{T_{0}} \sum_{t=1}^{T_{0}}\left({\varepsilon}_{1tj} -\sum_{W_{i}=0}\gamma^{\ast}_{i}{\varepsilon}_{itj}\right)^{2} +  2\left(\bar{\varepsilon}_{1\cdot j} -\sum_{W_{i}=0}\gamma^{\ast}_{i}\bar{\varepsilon}_{i\cdot j}\right)^{2}}\\
& \leq & \sqrt{\frac{2}{T_{0}} \sum_{t=1}^{T_{0}}\left({\varepsilon}_{1tj} -\sum_{W_{i}=0}\gamma^{\ast}_{i}{\varepsilon}_{itj}\right)^{2}} + \sqrt{ 2\left(\bar{\varepsilon}_{1\cdot j} -\sum_{W_{i}=0}\gamma^{\ast}_{i}\bar{\varepsilon}_{i\cdot j}\right)^{2}}.
\end{eqnarray*}
\end{proof}

\begin{lemma}[L2 norm of errors]\label{claim:better_qavg_errors}
 Suppose Assumptions~\ref{assumption:noiseless},~\ref{assumption:oracle}  and~\ref{assumption:SG} hold. For any $\delta>0$, we have the following bounds for the imbalance achieved by the oracle weights $\gamma^{\ast}$  
\begin{align}
\sqrt{\frac{1}{T_{0}}\frac{1}{K}\sum_{k=1}^{K}\sum_{t=1}^{T_{0}}\left(\varepsilon_{1tk} -\sum_{W_{i}=0}\gamma^{\ast}_{i}\varepsilon_{itk}\right)^{2}} & \leq 4\sigma\sqrt{1+\left\Vert \gamma^\ast\right\Vert _{2}^{2}}+\delta \label{eq:imbalance_errors_cat}\\
\sqrt{\frac{1}{T_{0}}\sum_{t=1}^{T_{0}}\left(\frac{1}{K}\sum_{k=1}^{K}\varepsilon_{1tk} -\sum_{W_{i}=0}\gamma^{\ast}_{i}\varepsilon_{itk}\right)^{2}} & \leq \frac{4\sigma\sqrt{1+\left\Vert \gamma^\ast\right\Vert _{2}^{2}}}{\sqrt{K}}+\delta \label{eq:imbalance_errors_avg}
\end{align}
with probability at least $1-2\exp\left(-\frac{T_0K\delta^2}{2\sigma^2(1+\left\Vert \gamma^\ast\right\Vert _{2}^{2})}\right)$.

Similarly, with probability at least $1-2\exp\left(-\frac{T_0\delta^2}{2\sigma^2(1+\left\Vert \gamma^\ast\right\Vert _{2}^{2})}\right)$, we have the following bounds for the separate imbalance achieved by the oracle weights $\gamma^{\ast}$  
\begin{align}
\sqrt{\frac{1}{T_{0}} \sum_{t=1}^{T_{0}}\left(\varepsilon_{1tj} -\sum_{W_{i}=0}\gamma^{\ast}_{i}\varepsilon_{itj}\right)^{2}} & \leq 4\sigma\sqrt{1+\left\Vert \gamma^\ast\right\Vert _{2}^{2}}+\delta \label{eq:imbalance_errors_sep}
\end{align}

\end{lemma}

\begin{proof}[Proof of Lemma~\ref{claim:better_qavg_errors}]

For the bound in~\eqref{eq:imbalance_errors_cat}, note that $\varepsilon_{1tk}-\sum_{W_{i}=0}\gamma_{i}^{\ast}\varepsilon_{1ik}$
is independent
across $t$ and $k$, and sub-Gaussian with scale parameter  $\sigma\sqrt{1+\left\Vert \gamma^\ast\right\Vert _{2}^{2}}$. 
Via a discretization argument from \cite{wainright2018high}[Ch.5], we can bound the LHS of ~\eqref{eq:imbalance_errors_cat}, a scaled  $L^{2}$ norm of a $(T_{0}K)\times1$ sub-Gaussian vector.  With probability at least $1-2\exp\left(-\frac{\delta^{2}}{2\sigma^2(1+\left\Vert \gamma^\ast\right\Vert _{2}^{2})}\right)$,
we have 
\begin{align*}
\sqrt{\frac{1}{T_{0}}\frac{1}{K}\sum_{k=1}^{K}\sum_{t=1}^{T_{0}}\left(\varepsilon_{1tk} -\sum_{W_{i}=0}\gamma^{\ast}_{i}\varepsilon_{itk}\right)^{2}} & \leq\frac{1}{\sqrt{T_{0}K}}\left(2\sigma\sqrt{1+\left\Vert \gamma^\ast\right\Vert _{2}^{2}}\sqrt{\log2+T_{0}K\log5}+\delta\right)\\
 & \leq 4\sigma\sqrt{1+\left\Vert \gamma^\ast\right\Vert _{2}^{2}}+\frac{1}{\sqrt{T_{0}K}}\delta
\end{align*} 
where we use the inequality $\log2+N\log5\leq 4N$ for positive $N$. 

For the bound in~\eqref{eq:imbalance_errors_avg}, note that  each $\bar{\varepsilon}_{1t}-\sum_{W_{i}=0}\gamma_{i}^{\ast}\bar{\varepsilon}_{1i}$
is independent
across $t$, and sub-Gaussian with scale parameter $\sigma/\sqrt{K}\sqrt{1+\left\Vert \gamma^\ast\right\Vert _{2}^{2}}$. we can similarly bound the LHS of ~\eqref{eq:imbalance_errors_avg}, a scaled  $L^{2}$ norm of a $(T_{0})\times1$ sub-Gaussian vector. With probability at least $1-2\exp\left(-\frac{K\delta^{2}}{2\sigma^2(1+\left\Vert \gamma^\ast\right\Vert _{2}^{2})}\right)$,
\begin{align*}
\sqrt{\frac{1}{T_{0}}\sum_{t=1}^{T_{0}}\left(\frac{1}{K}\sum_{k=1}^{K}\varepsilon_{1tk} -\sum_{W_{i}=0}\gamma^{\ast}_{i}\varepsilon_{itk}\right)^{2}} & \leq\frac{1}{\sqrt{T_{0}}}\left(2\frac{\sigma\sqrt{1+\left\Vert \gamma^\ast\right\Vert _{2}^{2}}}{\sqrt{K}}\sqrt{\log2+T_{0}\log5}+\delta\right)\\
 & \leq \frac{4\sigma\sqrt{1+\left\Vert \gamma^\ast\right\Vert _{2}^{2}}}{\sqrt{K}}+\frac{1}{\sqrt{T_{0}}}\delta
\end{align*}

Setting $\delta=\delta\sqrt{T_0K}$ for the tail bound of ~\eqref{eq:imbalance_errors_cat}, and $\delta=\delta\sqrt{T_0}$ for the tail bound of  ~\eqref{eq:imbalance_errors_avg}, we have the claimed result.

Finally for~\eqref{eq:imbalance_errors_sep}, we have a scaled  $L^{2}$ norm of a $(T_{0})\times1$ sub-Gaussian vector, each with a scale parameter  $\sigma\sqrt{1+\left\Vert \gamma^\ast\right\Vert _{2}^{2}}$.  Following a similar argument as above, we have with probability at least $1-2\exp\left(-\frac{\delta^{2}}{2\sigma^2(1+\left\Vert \gamma^\ast\right\Vert _{2}^{2})}\right)$,
we have 
\begin{align*}
\sqrt{\frac{1}{T_{0}} \sum_{t=1}^{T_{0}}\left(\varepsilon_{1tj} -\sum_{W_{i}=0}\gamma^{\ast}_{i}\varepsilon_{itj}\right)^{2}} & \leq\frac{1}{\sqrt{T_{0}}}\left(2\sigma\sqrt{1+\left\Vert \gamma^\ast\right\Vert _{2}^{2}}\sqrt{\log2+T_{0}\log5}+\delta\right)\\
 & \leq 4\sigma\sqrt{1+\left\Vert \gamma^\ast\right\Vert _{2}^{2}}+\frac{1}{\sqrt{T_{0}}}\delta
\end{align*} 
Setting $\delta=\delta\sqrt{T_0}$ for the tail bound of ~\eqref{eq:imbalance_errors_sep}, we have the claimed result.

\end{proof}

\begin{lemma}[Oracle imbalance]\label{claim:better_qavg}
Suppose Assumptions~\ref{assumption:noiseless},~\ref{assumption:oracle}  and~\ref{assumption:SG} hold. For any $\delta>0$, we have the following bounds for the imbalance achieved by the oracle weights $\gamma^{\ast}$:  \begin{enumerate}[label = (\roman*)]

    \item  if analyzing the separate imbalance 
\begin{align}
q^{sep}_k(\gamma^{\ast}) & \leq 4\sigma\sqrt{1+\left\Vert \gamma^\ast\right\Vert _{2}^{2}}+2\delta \label{eq:oracle_qsep}
\end{align}
with probability at least $1-4\exp\left(-\frac{T_0\delta^2}{2\sigma^2(1+\left\Vert \gamma^\ast\right\Vert _{2}^{2})}\right)$.
\item if analyzing the concatenated imbalance 
\begin{align}
q^{cat}(\gamma^{\ast}) & \leq 4\sigma\sqrt{1+\left\Vert \gamma^\ast\right\Vert _{2}^{2}}+2\delta +\frac{4\sigma\sqrt{1+\left\Vert \gamma^\ast\right\Vert _{2}^{2}}}{\sqrt{T_0}}\label{eq:oracle_qcat}
\end{align}
 with probability at least $1-4\exp\left(-\frac{T_0K\delta^2}{2\sigma^2(1+\left\Vert \gamma^\ast\right\Vert _{2}^{2})}\right)$.

\item if analyzing the average imbalance
\begin{align}
q^{avg}(\gamma^{\ast}) & \leq \frac{4\sigma\sqrt{1+\left\Vert \gamma^\ast\right\Vert _{2}^{2}}}{\sqrt{K}}+2\delta \label{eq:oracle_qavg}
\end{align}
 with probability at least $1-4\exp\left(-\frac{T_0K\delta^2}{2\sigma^2(1+\left\Vert \gamma^\ast\right\Vert _{2}^{2})}\right)$.
\end{enumerate}

\end{lemma}

\begin{proof}[Proof of Lemma~\ref{claim:better_qavg}]
First we apply Lemma~\ref{claim:L2_demeaned_errors} to derive a general upper bound.

For $q^{sep}_k(\gamma^{\ast})$, note that  each $\bar{\varepsilon}_{1\cdot k}-\sum_{W_{i}=0}\gamma_{i}^{\ast}\bar{\varepsilon}_{i\cdot k}$
is independent
across $k$, and sub-Gaussian with scale parameter $\frac{\sigma}{\sqrt{T_0}}\sqrt{1+\left\Vert \gamma^\ast\right\Vert _{2}^{2}}$.  Setting $\delta = \delta \slash \left(\frac{\sigma}{\sqrt{T_0}}\sqrt{1+\left\Vert \gamma^\ast\right\Vert _{2}^{2}}\right)$ in Lemma~\ref{lemma:post-treat}, we have that $\left|\bar{\varepsilon}_{1\cdot k} -\sum_{W_{i}=0}\gamma^{\ast}_{i}\bar{\varepsilon}_{i\cdot k} \right|$ is upper bounded by $\delta$ with probability at least $1-2\exp\left(-\frac{\delta^{2}T_0}{2\sigma^2(1+\left\Vert \gamma^*\right\Vert ^{2}_{2} )}\right)$.
 Applying the union bound, together with the bound in~\eqref{eq:imbalance_errors_sep} of Lemma~\ref{claim:better_qavg_errors}, we have the claimed  bound in~\eqref{eq:oracle_qsep}.
 
For $q^{cat}(\gamma^{cat})$, note that  each $\bar{\varepsilon}_{1\cdot k}-\sum_{W_{i}=0}\gamma_{i}^{\ast}\bar{\varepsilon}_{i\cdot k}$
is independent
across $k$, and sub-Gaussian with scale parameter $\frac{\sigma}{\sqrt{T_0}}\sqrt{1+\left\Vert \gamma^\ast\right\Vert _{2}^{2}}$. Using similar argument for the bound in ~\eqref{eq:imbalance_errors_sep} of Lemma~\ref{claim:better_qavg_errors}, we can bound the following scaled  $L^{2}$ norm of a $K\times1$ sub-Gaussian vector with probability at least $1-2\exp\left(-\frac{T_0K\delta^{2}}{2\sigma^2(1+\left\Vert \gamma^\ast\right\Vert _{2}^{2})}\right)$,
\begin{align*}
\sqrt{\frac{1}{K}\sum_{k=1}^{K}\left(\bar \varepsilon_{1\cdot k} -\sum_{W_{i}=0}\gamma^{\ast}_{i}\bar \varepsilon_{i\cdot k}\right)^{2}} & \leq \frac{4\sigma\sqrt{1+\left\Vert \gamma^\ast\right\Vert _{2}^{2}}}{\sqrt{T_{0}}}+\delta
\end{align*}
Applying the union bound, together with the bound in ~\eqref{eq:imbalance_errors_cat}, we have the claimed  bound in~\eqref{eq:oracle_qcat}.

For $q^{avg}(\gamma^{\ast})$, note that ${\frac{1}{K}\sum_{k=1}^K \bar \varepsilon}_{1\cdot k} -\sum_{W_{i}=0}\gamma^{\ast}_{i}\bar{\varepsilon}_{i\cdot k} $
is sub-Gaussian with scale parameter $\frac{\sigma}{\sqrt{KT_0}}\sqrt{1+\left\Vert \gamma^\ast\right\Vert _{2}^{2}}$.  Setting $\delta = \delta \slash \left(\frac{\sigma}{\sqrt{KT_0}}\sqrt{1+\left\Vert \gamma^\ast\right\Vert _{2}^{2}}\right)$ in 
 Lemma~\ref{lemma:post-treat}, we have that $\left|{\frac{1}{K}\sum_{k=1}^K \bar \varepsilon}_{1\cdot k} -\sum_{W_{i}=0}\gamma^{\ast}_{i}\bar{\varepsilon}_{i\cdot k} \right|$ is upper bounded by $\delta$ with probability at least $1-2\exp\left(-\frac{\delta^{2}KT_0}{2\sigma^2(1+\left\Vert \gamma^*\right\Vert ^{2}_{2})}\right)$. 
 Applying the union bound, together with the bound in ~\eqref{eq:imbalance_errors_avg} of Lemma~\ref{claim:better_qavg_errors}, we have the claimed  bound in~\eqref{eq:oracle_qavg}.

\end{proof}
\subsection{Error bounds for the approximation errors}

\begin{lemma}[Lemma A.4. of \cite{benmichael2021_ascm}]
\label{lem:approx_bdd}
 If $\xi_{i}$ are mean-zero sub-Gaussian random variables
with scale parameter $\bar{\omega}$, then for weights $\hat{\gamma}$
and any $\delta>0$, with probability at least $1-4\exp\left(-\frac{\delta^{2}}{2}\right)$,
we have 
\[
\left|\xi_{1}-\sum_{W_{i}=0}\hat{\gamma}_{i}\xi_{i}\right|\leq\delta\bar{\omega}+2\left\Vert \hat{\gamma}\right\Vert _{1}\bar{\omega}\left(\sqrt{\log2N_{0}}+\frac{\delta}{2}\right)=\bar{\omega}\left(2\left\Vert \hat{\gamma}\right\Vert _{1}\sqrt{\log2N_{0}}+\delta(1+\left\Vert \hat{\gamma}\right\Vert _{1})\right).
\]
\end{lemma}

\subsection{Error bounds for the post-treatment noise}
\label{sec:noise_bound}
\begin{lemma}\label{lemma:post-treat}
For weights independent of $\varepsilon_{iTj}$, under Assumption~\ref{assumption:noiseless} and~\ref{assumption:SG}, for any $\delta>0$ with probability at least $1-2\exp\left(-\frac{\delta^{2}}{2}\right)$,
we have 
\[
\left|\varepsilon_{1Tj}-\sum_{W_{i}=0}\hat{\gamma}_{i}\varepsilon_{iTj}\right|\leq\delta\sigma(1+\left\Vert \hat{\gamma}\right\Vert _{2}).
\]
\end{lemma}

\begin{proof}[Proof of Lemma \ref{lemma:post-treat}]
Since the weights are independent of $\varepsilon_{iTj}$, by sub-Gaussianity
and independence of $\varepsilon_{iTj}$, we see that $\varepsilon_{1Tj}-\sum_{W_{i}=0}\hat{\gamma}_{i}\varepsilon_{iTj}$
is sub-Gaussian with scale parameter $\sigma\sqrt{1+\left\Vert \hat{\gamma}\right\Vert _{2}^{2}}\leq\sigma(1+\left\Vert \hat{\gamma}\right\Vert _{2})$. Applying the Hoeffding's inequality, we obtained the claimed bound.
\end{proof}

\begin{lemma}\label{lemma:post-treat-demeaned}
For weights $\hat{\gamma}$
and any $\delta>0$, with probability at least $1-6\exp\left(-\frac{\delta^{2}}{2}\right)$,
we have 
\begin{align*}
\left|\dot \varepsilon_{1Tj}-\sum_{W_{i}=0}\hat{\gamma}_{i}\dot \varepsilon_{iTj}\right| & \leq\delta\sigma(1+\left\Vert \hat{\gamma}\right\Vert _{2})+\delta\frac{\sigma}{\sqrt{T_0}}+2\left\Vert \hat{\gamma}\right\Vert _{1}\frac{\sigma}{\sqrt{T_0}}\left(\sqrt{\log2N_{0}}+\frac{\delta}{2}\right)\\
& \leq (1+C)\delta\sigma\left(1+\frac{1}{\sqrt{T_0}}\right) + \frac{\sigma}{\sqrt{T_0}}\left(2C\sqrt{\log2N_{0}}\right) 
\end{align*}
\end{lemma}  

\begin{proof}[Proof of Lemma \ref{lemma:post-treat-demeaned}]
For the post-treatment noise, we have 
\begin{align*}
\left|\dot{\varepsilon}_{1Tj}-\sum_{W_{i}=0}\hat{\gamma}_{i}\dot{\varepsilon}_{iTj}\right| & =\left|\varepsilon_{1Tj}-\sum_{W_{i}=0}\hat{\gamma}_{i}\varepsilon_{iTj}+\sum_{W_{i}=0}\hat{\gamma}_{i}\bar{\varepsilon}_{i\cdot j}-\bar{\varepsilon}_{1\cdot j}\right|\\
 & \leq\left|\varepsilon_{1Tj}-\sum_{W_{i}=0}\hat{\gamma}_{i}\varepsilon_{iTj}\right|+\left|\sum_{W_{i}=0}\hat{\gamma}_{i}\bar{\varepsilon}_{i\cdot j}-\bar{\varepsilon}_{1\cdot j}\right|
\end{align*}
Lemma~\ref{lemma:post-treat} applies to the first term. However, for the second term, we note that $\bar{\varepsilon}_{i\cdot j}$ and $\hat{\gamma}_{i}$ are correlated, and Lemma~\ref{lem:approx_bdd} applies with a scale parameter of $\sigma/\sqrt{T_0}$. Applying  a union bound to the two terms, and note that $\left\Vert \hat{\gamma} \right\Vert _{2}\leq\left\Vert \hat{\gamma} \right\Vert _{1}=C$ by construction, we obtained the claimed bound.
\end{proof}
\section{Proofs of main results}\label{sec:proofs}
\begin{proof}[Proof for Proposition~\ref{claim:factor}]
For the system of linear  equations \eqref{eq:oracle} to have a solution, a necessary condition is that the matrix $\left(\begin{array}{cc}L &\mathbf{1}_N\end{array}\right)$ must have a reduced rank  less than $N$.  Furthermore, since all time effects are removed from $L$, the columns of $L$ are linearly independent from the one vector $\mathbf{1}_N$. Therefore, a necessary condition is for the rank of $L$ to be less than $N-1$. Since there exists a solution $\gamma^\ast$ to the system of linear  equations $L_{-1}'\gamma = L_1$, the Rouch{\'e}-Capelli theorem requires $rank(L_{-1})=rank(L).$
For the sufficiency, observe that appending the one vector to both $L_{-1}$ and $L$ increases their ranks by exactly one. Therefore, we still maintain $rank\left(\begin{array}{cc}L_{-1} &\mathbf{1}_{N-1}\end{array}\right)= rank\left(\begin{array}{cc}L &\mathbf{1}_N\end{array}\right)\leq N-1.$ The Rouch{\'e}-Capelli theorem then guarantees the existence of solutions to the system of linear  equations \eqref{eq:oracle}.
\end{proof}
\paragraph{Proof of Theorem  \ref{thm:error_bounds}.} The proof follows from Theorem~\ref{thm:sep_bdd}, ~\ref{thm:cat_bdd} and ~\ref{thm:avg_bdd}, separately proved below.

\begin{theorem}[Bound for separate weights]\label{thm:sep_bdd}
Suppose Assumptions~\ref{assumption:noiseless},~\ref{assumption:oracle},~\ref{assumption:SG} and ~\ref{assumption:singularity} hold. Then for any $\delta>0$, we have the following bound

\begin{align*}
  \left\vert Bias(\hat{\gamma_j}^{sep})\right\vert&\leq  \frac{r_jM^{2}}{\underbar{\ensuremath{\xi}}^{sep}}\left( (4\sigma(1+\|\gamma^\ast\|_2)+2\delta) +  \frac{\sigma\cdot(1+1/\sqrt{T_0})}{\sqrt{T_{0}}}\left( 2C\sqrt{\log2N_{0}} + (1+C)\delta\right) \right)
 \end{align*}
 with probability at least $1-8\exp\left(-\frac{\delta^{2}}{2}\right)-4\exp\left(-\frac{T_0\delta^2}{2\sigma^2(1+C^{2})}\right)$.
\end{theorem}
 
\begin{proof}[Proof of Theorem \ref{thm:sep_bdd}] As discussed in the main text, denote the projected factor value by $\omega_{tj}=\mu_{Tj}\cdot\left(\sum_{t=1}^{T_{0}}\mu_{tj}\mu'_{tj}\right)^{-1} \mu_{tj} $,  we can decompose the bias into the following two terms:

$$
\dot{L}_{1Tj}(0)-\sum_{W_{i}=0}\hat{\gamma}^{sep}_{i} \dot{L}_{iTj} =  \sum_{t=1}^{T_{0}}\omega_{tj}(\dot{Y}_{1tj}-\sum_{W_{i}=0}\hat{\gamma}^{sep}_{i}\dot{Y}_{itj})  -\sum_{t=1}^{T_{0}}\omega_{tj}(\dot{\varepsilon}_{1tj}-\sum_{W_{i}=0}\hat{\gamma}^{sep}_{i}\dot{\varepsilon}_{itj}) $$
 Next we derive the upper bound for the absolute value of each term.By Assumption~\ref{assumption:singularity}, for all $t$ we have $ (\omega_{tj})^{2}\leq \left(\frac{r_jM^{2}}{\underbar{\ensuremath{\xi}}^{sep}T_{0}}\right)^2$. Next we derive the upper bound for the absolute value of each term.

 To bound the bias due to imbalance, we apply the Cauchy-Schwarz
inequality:
\begin{align*}
 (R_{0}^{sep}) = & \sum_{t=1}^{T_{0}}\omega_{tj}(\dot{Y}_{1tj}-\sum_{W_{i}=0}\hat{\gamma}^{sep}_{i}\dot{Y}_{itj})
\leq\sqrt{\sum_{t=1}^{T_{0}}\omega_{tj}^{2}}\sqrt{\sum_{t=1}^{T_{0}}(\dot{Y}_{1tj}-\sum_{W_{i}=0}\hat{\gamma}^{sep}_{i}\dot{Y}_{itj})^{2}}\\
= & \sqrt{T_{0}}\sqrt{\sum_{t=1}^{T_{0}}\omega_{tj}^{2}}\sqrt{\frac{1}{T_{0}}\sum_{t=1}^{T_{0}}(\dot{Y}_{1tj}-\sum_{W_{i}=0}\hat{\gamma}^{sep}_{i}\dot{Y}_{itj})^{2}}\\
\leq & \sqrt{T_{0}}\sqrt{T_{0}\cdot\left(\frac{r_jM^{2}}{\underbar{\ensuremath{\xi}}^{sep}T_{0}}\right)^{2}}q^{sep}(\hat{\gamma}^{sep})
= (\underbar{\ensuremath{\xi}})^{-1} r_jM^{2}q^{sep}(\hat{\gamma}^{sep})\leq(\underbar{\ensuremath{\xi}}^{sep})^{-1} r_jM^{2}q^{sep}(\gamma^{\ast}).
\end{align*}
 Lemma~\ref{claim:better_qavg} derives a  high-probability upper bound for $q^{sep}(\gamma^{\ast})$, which gives an upper bound for $|R_{0}^{sep}|$.

For $|R_{1}^{sep}|$, set $\xi_{i}=\sum_{t=1}^{T_{0}}\omega_{tj}\varepsilon_{itj}$ and $\bar{\xi}_{i}=\bar \varepsilon_{i\cdot j}\sum_{t=1}^{T_{0}}\omega_{tj}$. We therefore have the upper bound $$|R_{1}^{sep}|=\left|\xi_{1}-\sum_{W_{i}=0}\hat{\gamma}_{i}\xi_{i} - \bar \xi_{1} + \sum_{W_{i}=0}\hat{\gamma}_{i}\bar \xi_{i}\right|\leq \left|\xi_{1}-\sum_{W_{i}=0}\hat{\gamma}_{i}\xi_{i}\right| +  \left|\bar \xi_{1}-\sum_{W_{i}=0}\hat{\gamma}_{i}\bar \xi_{i}\right|$$Furthermore,   the weighted sum $\xi_{i}$ is sub-Gaussian with a
scale parameter $\frac{\sigma}{\sqrt{T_{0}}}\frac{r_jM^{2}}{\underbar{\ensuremath{\xi}}^{sep}}$, and $\bar{\xi}_{i}$ is sub-Gaussian with a
scale parameter $\frac{\sigma}{T_{0}}\frac{r_jM^{2}}{\underbar{\ensuremath{\xi}}^{sep}}$.  We apply Lemma~\ref{lem:approx_bdd} to both terms with the union bound.

Combining the  probabilities with the union bound gives the result with probability at least $1-8\exp\left(-\frac{\delta^{2}}{2}\right)-4\exp\left(-\frac{T_0\delta^2}{2\sigma^2(1+\left\Vert \gamma^\ast\right\Vert _{2}^{2})}\right)$, the bias is upper bounded by $$ \frac{r_jM^{2}}{\underbar{\ensuremath{\xi}}^{sep}}\left( 4\sigma\sqrt{1+\left\Vert \gamma^\ast\right\Vert _{2}^{2}}+2\delta +  \frac{\sigma\cdot(1+1/\sqrt{T_0})}{\sqrt{T_{0}}}\left( 2\left\Vert \hat{\gamma}^{sep}\right\Vert _{1}\sqrt{\log2N_{0}} + \delta(1+\left\Vert \hat{\gamma}^{sep}\right\Vert _{1})\right) \right).
 $$

 We then note that $\left\Vert \hat{\gamma}^{sep}\right\Vert _{1}=C$ by construction and
 $\sqrt{1+\left\Vert \gamma^\ast\right\Vert _{2}^{2}}\leq  1+\left\Vert \gamma^\ast\right\Vert  _{2}$.
 \end{proof}

\begin{theorem}[Bound for concatenated weights]\label{thm:cat_bdd}
Suppose Assumptions~\ref{assumption:noiseless},~\ref{assumption:oracle},  and~\ref{assumption:SG}  and ~\ref{assumption:singularity}. 
Then for any $\delta>0$, we have the following bound
\begin{align*}
 \left\vert Bias(\hat{\gamma}^{cat})\right\vert &\leq  \frac{rM^{2}}{\underbar{\ensuremath{\xi}}^{cat}}\left( (4\sigma(1+1/\sqrt{T_0})(1+\|\gamma^\ast\|_2)+2\delta ) +  \frac{\sigma\cdot(1+1/\sqrt{T_0})}{\sqrt{T_{0}K}}\left( 2C\sqrt{\log2N_{0}} + (1+C)\delta\right) \right)
 \end{align*}
 with probability at least $1-8\exp\left(-\frac{\delta^{2}}{2}\right)-4\exp\left(-\frac{T_0K\delta^2}{2\sigma^2(1+C^{2})}\right)$.
\end{theorem}

\begin{proof}[Proof of Theorem \ref{thm:cat_bdd}]
 As discussed in the main text, denote the projected factor value to be $\omega_{tk}=\mu_{Tj}\cdot\left(\sum_{k=1}^{K}\sum_{t=1}^{T_{0}}\mu_{tk}\mu_{tk}'\right)^{-1}\mu_{tk}$, we can decompose the bias  into the following two terms $R_{0}^{cat}$ and $R_{1}^{cat}$:

$$ L_{1Tj}(0)-\sum_{W_{i}=0}\hat{\gamma}^{cat}_{i} L_{iTj} =  \sum_{k=1}^{K}\sum_{t=1}^{T_{0}}\omega_{tk}(Y_{1tk}-\sum_{W_{i}=0}\hat{\gamma}^{cat}_{i}Y_{itk})  - \sum_{k=1}^{K}\sum_{t=1}^{T_{0}}\omega_{tk}(\varepsilon_{1tk}-\sum_{W_{i}=0}\hat{\gamma}^{cat}_{i}\varepsilon_{itk}). 
$$  
The rest of the proof therefore mimics the proof of Theorem~\ref{thm:sep_bdd}, or can be found in the prior version of this paper \citep{sun2023usingmultipleoutcomesimprove}.
\end{proof}

\begin{theorem}[Bound for average weights]\label{thm:avg_bdd}
Suppose Assumptions~\ref{assumption:noiseless},~\ref{assumption:oracle}, ~\ref{assumption:SG} and~\ref{assumption:singularity} hold. Then for any $\delta>0$, we have the following bound
\begin{align*}
 \left\vert Bias(\hat{\gamma}^{avg})\right\vert&\leq  \frac{rM^{2}}{\underbar{\ensuremath{\xi}}^{avg}}\left( (\frac{4\sigma}{\sqrt{K}}(1+\|\gamma^\ast\|_2)+2\delta) +  \frac{\sigma\cdot(1+1/\sqrt{T_0})}{\sqrt{T_{0}K}}\left( 2C\sqrt{\log2N_{0}} + (1+C)\delta\right) \right)
 \end{align*}
 with probability at least $1-8\exp\left(-\frac{\delta^{2}}{2}\right)-4\exp\left(-\frac{T_0K\delta^2}{2\sigma^2(1+C^{2})}\right).$ 
\end{theorem}

 \begin{proof}[Proof of Theorem \ref{thm:avg_bdd}] 
 Denote the average outcome $\bar{Y}_{it}=\frac{1}{K}\sum_{k=1}^{K}Y_{itk}$
and similarly $\bar{\mu}_{t}=\frac{1}{K}\sum_{k=1}^{K}\mu_{tk}$. Denote the projected average factor value to be $\omega_{tj}=\mu_{Tj}\cdot\left(\sum_{t=1}^{T_0}\left(\bar{\mu}_{t}\right)\left(\bar{\mu}_{t}\right)'\right)^{-1}\bar{\mu}_{t}$ we can decompose the bias into the following two terms $R_{0}^{avg}$ and $R_{1}^{avg}$:
 
$$L_{1Tj}(0)-\sum_{W_{i}=0}\hat{\gamma}^{avg}_{i} L_{iTj} = \sum_{t=1}^{T_{0}}\omega_{tj} (\bar{Y}_{1t}-\sum_{W_{i}=0}\hat{\gamma}^{avg}_{i}\bar{Y}_{it})  -\sum_{t=1}^{T_{0}}\omega_{tj}(\bar{\varepsilon}_{1t}-\sum_{W_{i}=0}\hat{\gamma}^{avg}_{i}\bar{\varepsilon}_{it}).$$

The rest of the proof therefore mimics the proof of Theorem~\ref{thm:sep_bdd}, or can be found in the prior version of this paper \citep{sun2023usingmultipleoutcomesimprove}.

\end{proof}

\begin{lemma}[combined weights]\label{lem:combined imbalance}
   Suppose there exists $\nu^\ast \in [0,1]$ such that for $\hat\gamma^{\com} \in\arg\min_{\gamma\in \mathcal C} \nu^\ast q^{avg}(\gamma) + (1-\nu^\ast) q^{cat}(\gamma)$, we have $q^{avg}(\gamma^{\com}) \leq q^{avg}(\gamma^\ast) $ and $q^{cat}(\gamma^{\com}) \leq q^{cat}(\gamma^\ast) $ almost surely. For any $\delta>0$, let $\tilde \sigma=\left( 2C\sqrt{\log2N_{0}} + (1+C)\delta\right)(1+1/\sqrt{T_0})\sigma$,  with probability at least $1-8\exp\left(-\frac{\delta^{2}}{2}\right)-4\exp\left(-\frac{T_0K\delta^2}{2\sigma^2(1+C^2)}\right)$, the absolute bias satisfies the bound, 
 
\begin{align*}
\left| \text{Bias}(\hat{\gamma}^{\com}) \right| \leq \min & \left\{ \frac{rM^{2}}{\underline{\xi}^{cat}} \left( 4(1+C)\sigma + 2\delta + \frac{\tilde{\sigma}}{\sqrt{T_{0}K}} \right),     \frac{rM^{2}}{\underline{\xi}^{cat}} \left( \frac{4(1+C)\sigma}{\sqrt{K}} + 2\delta + \frac{\tilde{\sigma}}{\sqrt{T_{0}K}} \right) \right\}.  
\end{align*}  
Furthermore, under Assumptions~\ref{a:conformal} and~\ref{a:conformal_avg}, the conformal inference procedure outlined in Section~\ref{sec:inference} is valid for $\gamma^\com$.
\end{lemma}
\begin{proof}
    Since  $q^{avg}(\gamma^{\com}) \leq q^{avg}(\gamma^\ast) $ and $q^{cat}(\gamma^{\com}) \leq q^{cat}(\gamma^\ast) $ almost surely, either of the two bias bounds  stated in Theorem~\ref{thm:cat_bdd} and~\ref{thm:avg_bdd} is a valid upper bound for the estimate based on the combined weights 
 $\hat\gamma^{\com}$. We may therefore take the minimum of the two bounds to bound $|Bias(\hat\gamma^{\com})|$.

 Furthermore, by the assumption of $q^{avg}(\gamma^{\com}) \leq q^{avg}(\gamma^\ast) $, we have
    \[
        \frac{1}{T}\sum_{t=1}^T\left(\dot{\bar{Y}}_{1t\cdot} - \sum_{W_i=0} \dot{\bar{Y}}_{it\cdot}\hat{\gamma}_i^{\com}\right)^2\leq \frac{1}{T}\sum_{t=1}^T\left(\dot{\bar{Y}}_{1t\cdot} - \sum_{W_i=0} \dot{\bar{Y}}_{it\cdot} \gamma^\ast_i\right)^2.  
    \]
    The proof of Lemma~\ref{lem:conformal}, which is based on the same inequality for $\gamma^{avg}$, proceeds in a similar fashion, thereby establishing the validity of the conformal inference procedure outlined in Section~\ref{sec:inference} for $\gamma^\com$ as well.
\end{proof}

\section{Simulations}
\label{sec:sims}
We conduct a Monte Carlo study to further inspect the behavior of SCM estimators based on separate, concatenated, and average weights.

First, to focus on key ideas, we consider a simple model of the $k$\textsuperscript{th} outcome under control, 
\begin{equation}
Y_{itk}(0)=\phi_{i}\mu_{t}+\varepsilon_{itk},\label{eq:common}
\end{equation}
where $\phi_{i}$ is a scalar and $\varepsilon_{itk}\sim\mathcal{N}(0,1)$. Here multiple outcomes are repeated independent measurements of the same underlying model component that consists of a single latent factor.
We consider four settings for the number of pre-treatment time periods $T_0$ and outcomes $K$: (i) $T_0 = 10, K = 4$; (ii) $T_0 = 10, K = 10$; (iii) $T_0 = 40, K = 4$; (iv) $T_0 = 40, K = 10$.

 The factor loadings $\phi_i$ are evenly spaced over the interval $[1,5]$ for $i=1,\dots,50$. Similar to \cite{Ferman_Properties_2021}, we set the treated unit to be the unit with the second largest factor loading.
This choice injects selection of the treated unit based on the factor loadings, so that a simple difference in means would be biased.
It also guarantees the existence of oracle weights that solve  $\phi_{1}-\sum_{W_{i}=0}\gamma_{i}^{\ast}\phi_{i}=0$. We set the factor values $\mu_t$ to be evenly spaced over the interval $[0.5,1]$ for $t=1,\dots,T_0+1$, reflecting an upward time trend.
\begin{figure}
    \centering
    \begin{subfigure}{0.48\textwidth}
        \centering
        \includegraphics[width=\textwidth, trim=0cm 2cm 0cm 2cm]{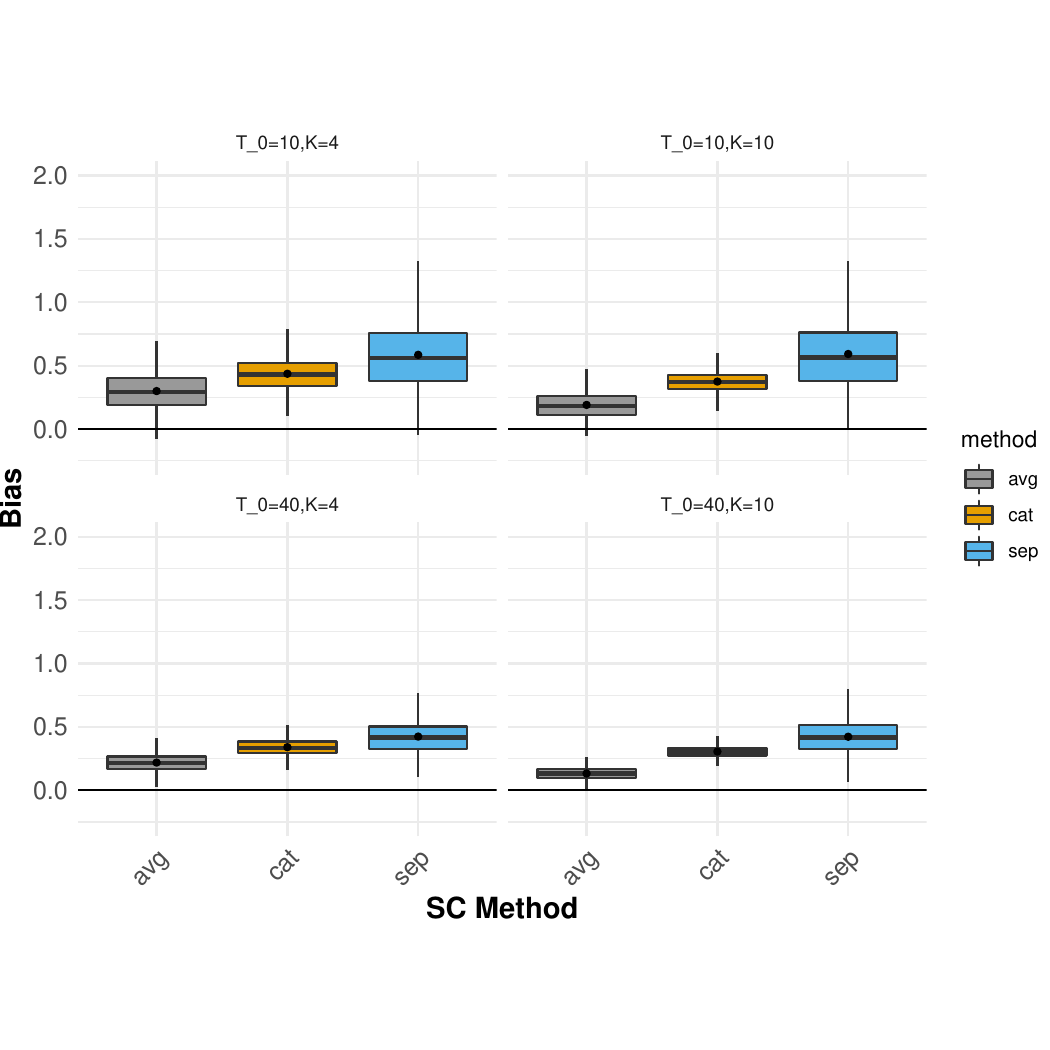}
        \caption{}
        \label{fig:bias_sims}
    \end{subfigure}
    \hfill
    \begin{subfigure}{0.48\textwidth}
        \centering
        \includegraphics[width=\textwidth, trim=0cm 2cm 0cm 2cm]{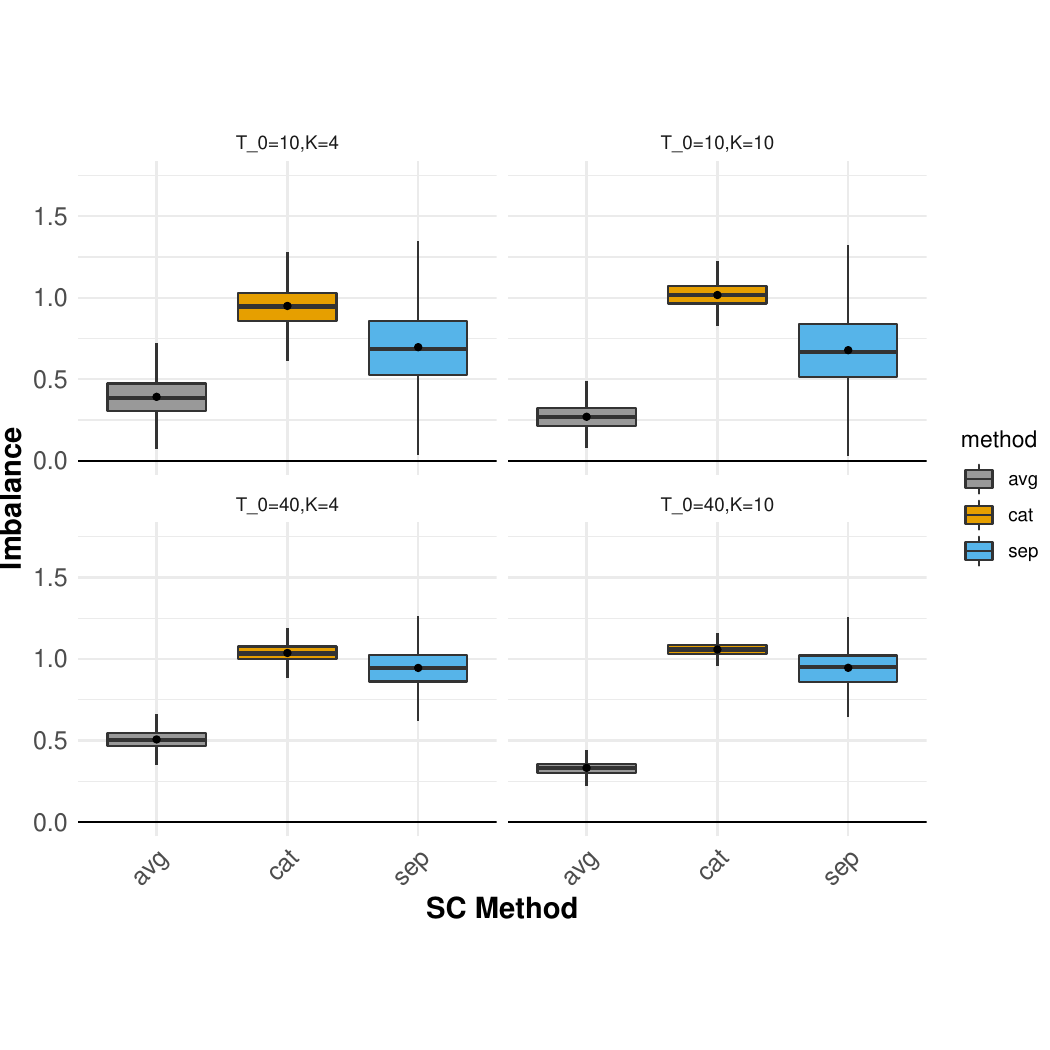}
        \caption{}
        \label{fig:imbalance_sims}
    \end{subfigure}
    \caption{Box plots of bias and imbalance using separate SCM, concatenated SCM, and average SCM over 1000 simulations.}
    \label{fig:bias_and_imbalance}
\end{figure}

Figure~\ref{fig:bias_sims} compares the distribution of the bias for estimating the treatment effect on the first outcome under different weighting estimators.
Consistent with Theorem~\ref{thm:error_bounds}, Figure~\ref{fig:bias_sims} illustrates that, relative to separate weights, the concatenated and average weights reduce bias in settings with multiple outcomes.
We also see that, as expected, the average weights have smaller average bias than the concatenated weights.

To further inspect this, Figure~\ref{fig:imbalance_sims} contrasts the imbalance for each type of weight with the corresponding objective functions.
First, the concatenated weights have slightly greater imbalance than the separate weights, highlighting the difficulty in achieving good pre-treatment fit on all outcomes simultaneously relative to good pre-treatment fit for a single outcome alone.
However, the average bias for the concatenated weights is still smaller than for the separate weights, showing that the reduction in overfitting by concatenating  outweighs the slight reduction in pre-treatment fit.
Second, the average weights have much better pre-treatment fit than either alternative, with the fit improving as $K$ increases.
As Figure~\ref{fig:bias_sims} shows, this leads to further bias reduction, 
consistent with Theorem~\ref{thm:error_bounds} and the intuition from Table~\ref{tab:bound_rates}.

Second, we examine how the presence of idiosyncratic factors influences the performance of various estimators. For $\rho\in[0,1]$ where $\rho$ adjusts the importance of the common model component relative to the idiosyncratic model components, we generate control outcomes  from
\[
Y_{itk}(0)=\rho \phi_{i}\mu_{t}+(1-\rho) \phi_{ik}\mu_{tk}+\varepsilon_{itk}.
\]
Here, $\phi_{i}$ and $\mu_{t}$ represent the common model component, as previously defined in~\eqref{eq:common}. We set $\phi_{i1}=\phi_{i}$ and $\mu_{t1}=\mu_{t}$, ensuring that the first outcome is generated exactly as previously and that the performance of the separate SCM is held fixed. To introduce the idiosyncratic model components, for each $k=2,\dots,K$, we independently generate $\phi_{ik}$ from a standard normal and $\mu_{tk}$ from an autoregressive process with an autoregressive coefficient of 0.5. We then rescale $\phi_{ik}$ and $\mu_{tk}$ to match the range of $\phi_{i1}$ and $\mu_{t1}$, respectively. We set $\phi_{1k} = \sum_{W_i=0}\gamma_i^\ast
\phi_{ik}$ using the same oracle weights $\gamma_i^\ast$ as before to maintain the existence of the oracle weights. Importantly, these outcome-specific factors and loadings are  generated independently of each other to  reflect the idiosyncratic components of the overall factor structure.

\begin{figure}[htb]
    \centering
    \begin{subfigure}{0.48\textwidth}
        \centering
        \includegraphics[width=\textwidth, trim=0cm 2cm 0cm 2cm]{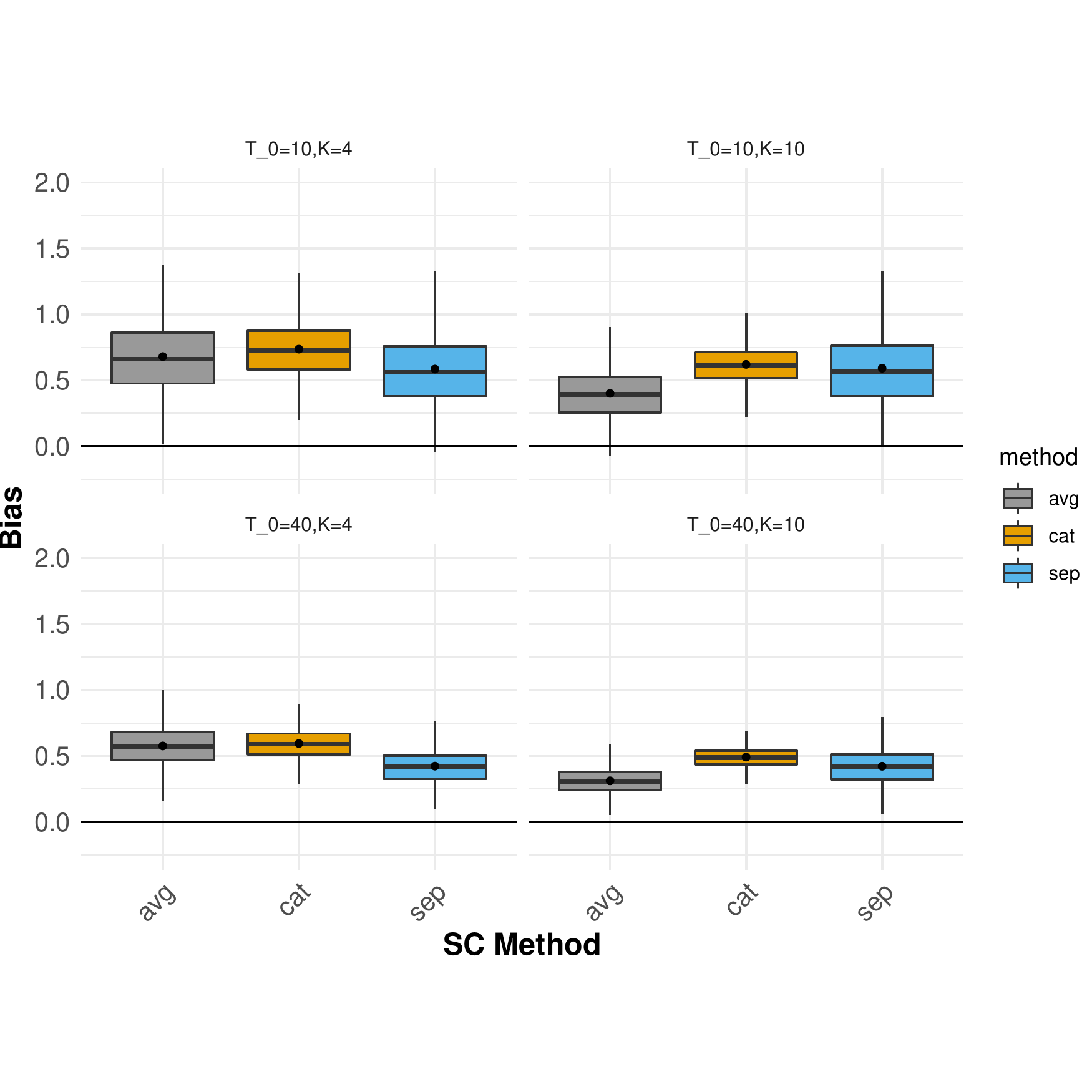}
        \caption{Common + idiosyncratic factors; $\rho=0.5$}
        \label{fig:bias_sims_rho05}
    \end{subfigure}
    \hfill
    \begin{subfigure}{0.48\textwidth}
        \centering
        \includegraphics[width=\textwidth, trim=0cm 2cm 0cm 2cm]{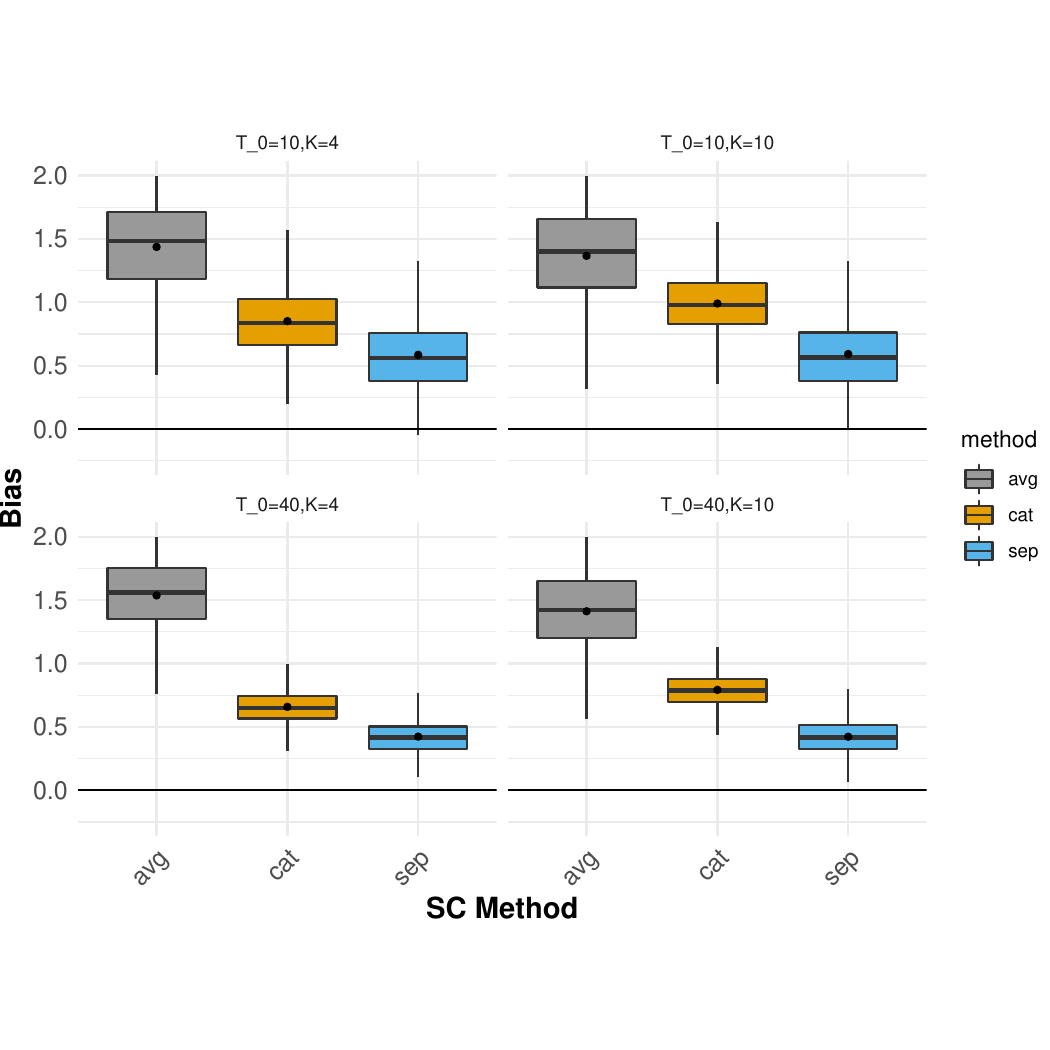}
        \caption{Only idiosyncratic factors; $\rho=0$}
        \label{fig:bias_sims_rho0}
    \end{subfigure}
    \caption{Box plots of bias using separate SCM, concatenated SCM, and average SCM over 1000 simulations.}
    \label{fig:bias_idiosyncratic}
\end{figure}

To assess the impact of idiosyncratic factors on performance, we vary $\rho\in[0,1]$. When $\rho=1$, we revert to the common model in~\eqref{eq:common}.  For $\rho\in(0,1)$,  the model combines common and idiosyncratic components as described in~\eqref{eq:common and idiosyncratic}, where the  components are  correlated across $k=1,\dots,K$ through the common component $\phi_i\mu_t$.  In this case, averaging outcomes still reduces bias in the synthetic control as illustrated in Figure~\ref{fig:bias_sims_rho05}. Though the bias reduction requires a larger number of outcomes for the common component to dominate. However, when $\rho=0$, the outcomes are generated by idiosyncratic model components, creating an adversarial DGP for concatenated and average SCM as illustrated in Figure~\ref{fig:bias_sims_rho0}. A significant presence of idiosyncratic factors can be partly assessed in practice by spectral analysis of the observed outcome. For separate SCM, the top singular vector on average captures 86\% of the total variation in the $N\times {T_0}$ pre-treatment data matrix $Y_{it1}$. However, for concatenated SCM, the top singular vector captures only 45\% of the total variation in the $N\times (T_0  K)$ pre-treatment data matrix $Y_{itk}$. This suggests the presence of idiosyncratic factors as different outcomes are captured by different components. Average SCM exacerbates the bias  because averaging idiosyncratic factors across outcomes reduces overall variation, thereby violating Assumption~\ref{assumption:singularity}. 

A loss of signal-to-noise can also be  assessed in practice based on the condition number of the averaged outcome, which is the ratio of the largest  to the smallest singular value. Since averaging reduces the variation of the noise, average SCM should increase the condition number compared to separate SCM if there is a strong signal.  With only common factors ($\rho=1$), the condition number of average SCM is on average 227\% larger than that of separate SCM, indicating strong signal. 
With only idiosyncratic factors ($\rho=0$), this increase reduces to only 25\%, suggesting that average SCM offers little advantage over separate SCM in such cases.

\section{Additional Details for Flint Water Crisis Study}
\label{appendix:flint}

\paragraph{Data processing.} Math and reading achievement are measured via the annual state-administered educational assessments for grades 3-8, and are standardized at the grade-subject-year level. Special needs status is measured as the percent of students with a qualified special educational need. Attendance is in percent of days attended. The math, reading, and special needs series begin in 2007; daily attendance begins in 2009. Note that \citet{trejo_psychosocial_2021} also use 2006 data for special needs; we start our data series in 2007 to have multiple outcomes available for averaging, dropping attendance from the average for 2007 and 2008.
Finally, when averaging, we further standardize each outcome series using the series pre-treatment standard deviation.

\begin{table}[ht]
\centering
\begin{tabular}{rlr}
  \hline
 & District Name & Combined SCM Weight \\ 
  \hline
1 & Dowagiac Union School District & 0.19 \\ 
  2 & Oak Park School District & 0.18 \\ 
  3 & Lincoln Consolidated School District & 0.15 \\ 
  4 & Hamtramck School District & 0.14 \\ 
  5 & Houghton Lake Community Schools & 0.12 \\ 
  6 & Whittemore-Prescott Area Schools & 0.08 \\ 
  7 & River Rouge School District & 0.04 \\ 
  8 & Van Buren Public Schools & 0.04 \\ 
  9 & Beecher Community School District & 0.04 \\ 
  10 & Bloomingdale Public School District & 0.01 \\ 
   \hline
\end{tabular}
\caption{Synthetic control weights on Michigan districts combining Student Attendance, Special Needs, Reading Achievement, and Math Achievement outcomes using the combined objective with the heuristic choice of $\nu$. All districts not included recieved a weight of less than 0.001.}
\label{tab:combined_weights}
\end{table}

\begin{figure}[ht]
  \begin{centering}
        \includegraphics[width=\textwidth]{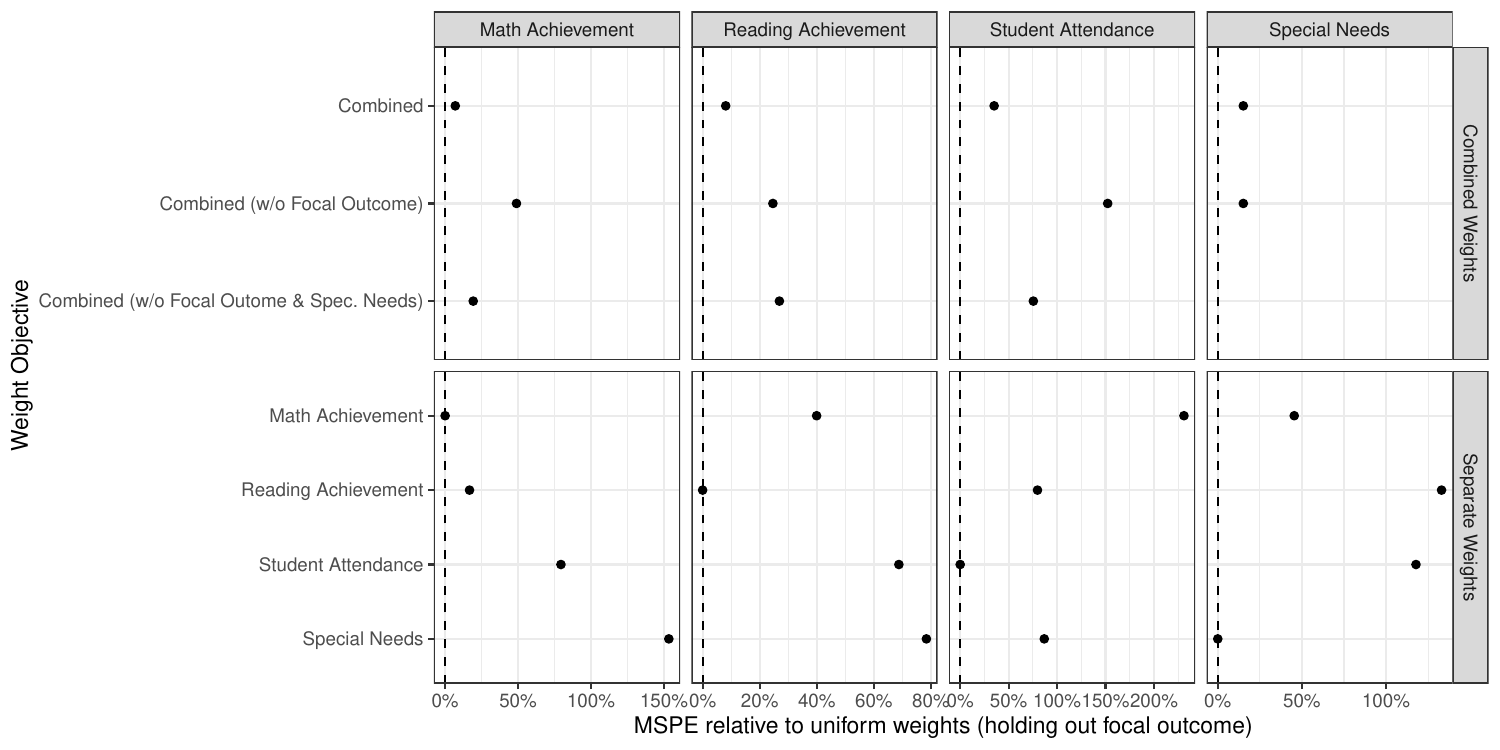}
  \par\end{centering}
  \caption{Mean Square Prediction Error (MSPE) of synthetic control, relative to the MSPE using uniform weights, for each outcome, fitting a synthetic control (i) separately on each outcome, (ii) combining all outcomes, (iii) combining all outcomes except for the focal outcome, and (iv) combining all outcomes except for the focal outcome and special needs. All combined objectives use an equal weight of $\nu = 0.5$ between the averaged and concatenated objectives. }\label{fig:pre_fit}

  \end{figure}

\begin{figure}[ht]
\begin{centering}
     \includegraphics[width=0.9\textwidth]{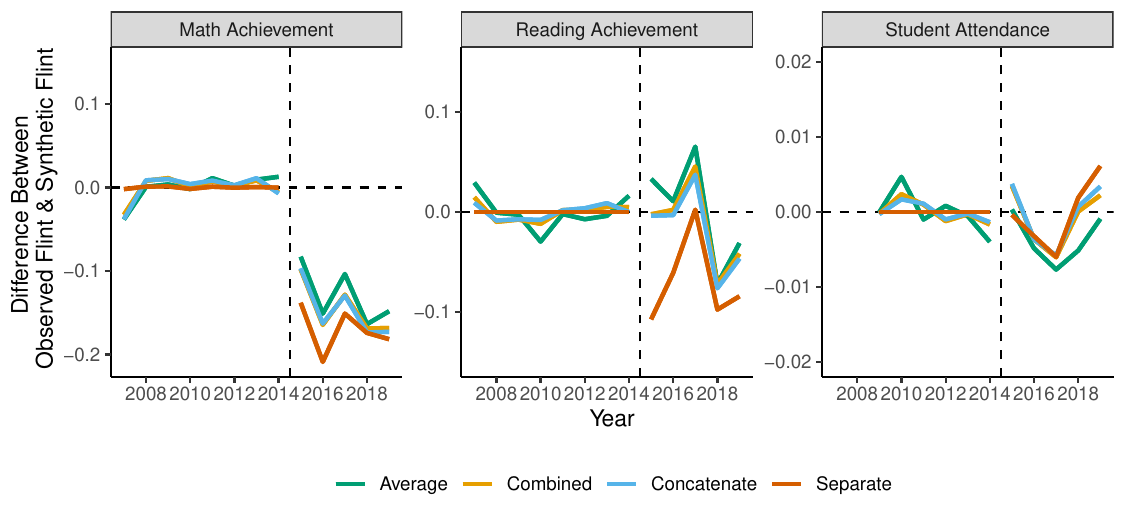}
\par\end{centering}
\caption{Point estimates for the effect of the Flint water crisis using SCM, concatenate SCM, and average SCM, without including special needs. Each outcome is standardized using the series pre-treatment standard deviation.}\label{fig:flint_nosped}
\end{figure}

\paragraph{Inference.} 
We use the conformal inference procedure outlined in Appendix \ref{sec:inference} to assess uncertainty in our reanalysis of \cite{trejo_psychosocial_2021}, with the caveat that the number of pre-treatment periods is only slightly larger than the number of post-treatment periods in this application.
We first test the null hypothesis of no effect on any outcomes in each time period, using the combined approach and i.i.d. permutations; this yields $p$-values for of 0.55, 0.11, 0.1, 0.24, 0.22 for 2015 to 2019. We then test the joint null hypothesis of no effect on any outcomes in any time period via a conformal inference procedure using all post-treatment time periods; here we find strong evidence against the null of no effect whatsoever, with $p = 0.035$. 
We also explore the sensitivity of the estimates by varying the combination between the concatenate and the average objective. Figure~\ref{fig:p_val_plot} illustrates that the results remain statistically significant at the 10\% level across a diverse range of combined weights $\nu$ --- including for the concatenated and average weights --- highlighting the robustness of the estimated negative impact.

\begin{figure}[ht]
\begin{centering}
     \includegraphics[width=0.5\textwidth]{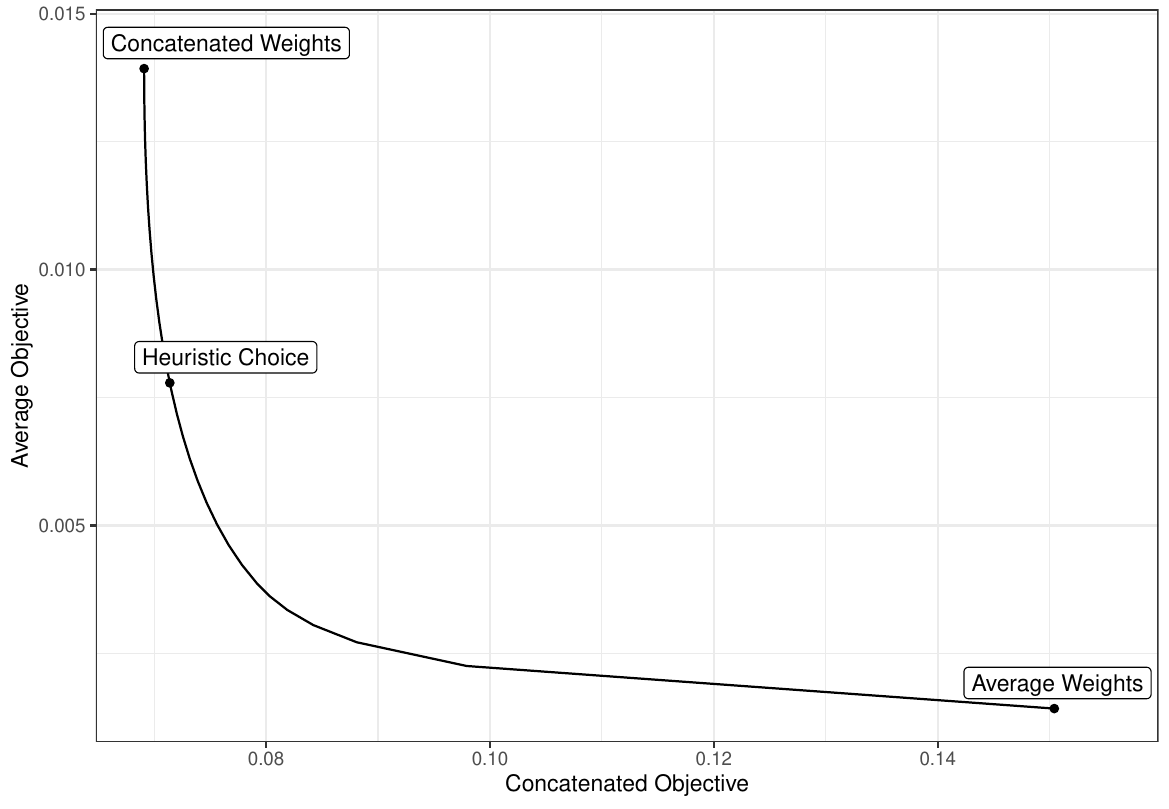}
\par\end{centering}
\caption{Frontier plot.}\label{fig:frontier}
\end{figure}

\begin{figure}[ht]
  \begin{centering}
       \includegraphics[width=0.5\textwidth]{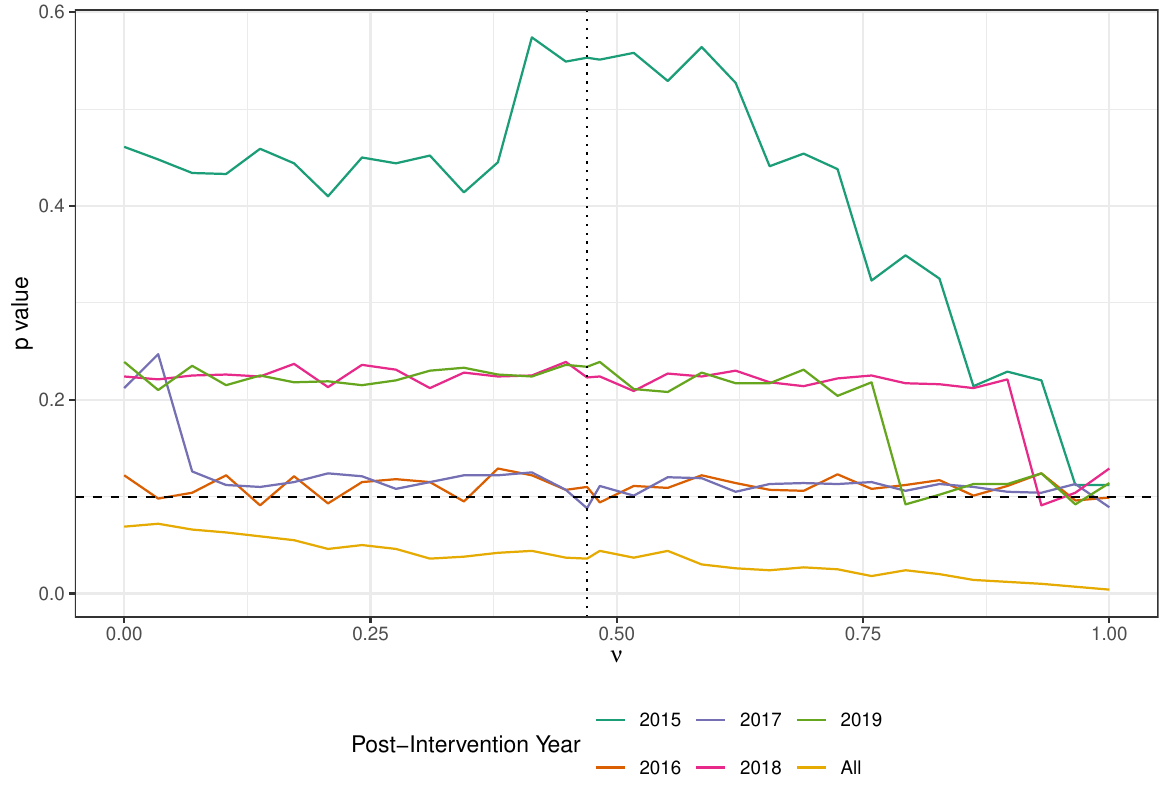}
  \par\end{centering}
  \caption{$p$ value of null of no effect on any outcome vs hyper-parameter $\nu$.}\label{fig:p_val_plot}
  \end{figure}

\end{document}